\documentclass[11pt]{article}

\usepackage{authblk}
\usepackage[T1]{fontenc}
\usepackage[utf8]{inputenc}
\usepackage{lmodern,microtype}
\usepackage{amsmath,amssymb,amsthm}
\usepackage[top=2.5cm,bottom=2.5cm,left=2.5cm,right=2.5cm]{geometry}
\usepackage{graphicx,booktabs,array,tabularx}
\usepackage{tikz}
\usetikzlibrary{positioning,arrows.meta,fit,backgrounds}
\usepackage{xcolor}
\usepackage[numbers,square,sort&compress]{natbib}
\usepackage{hyperref}
\hypersetup{
    colorlinks=true,
    linkcolor=blue,
    citecolor=blue,
    urlcolor=blue
    }
\usepackage{bookmark}
\usepackage{enumitem}
\setlist{itemsep=2pt,topsep=5pt}
\newtheorem{theorem}{Theorem}
\newtheorem{lemma}[theorem]{Lemma}

\newtheorem{corollary}[theorem]{Corollary}
\theoremstyle{definition}
\newtheorem{definition}[theorem]{Definition}
\theoremstyle{remark}

\newcommand{\SEP}{\operatorname{SEP}}
\newcommand{\STSP}{\operatorname{STSP}}
\newcommand{\CP}{\operatorname{CP}}
\newcommand{\NP}{\mathsf{NP}}
\newcommand{\PP}{\mathsf{P}}
\newcommand{\coNP}{\mathsf{coNP}}

\newcommand{\cC}{\mathcal{C}}

\newcommand{\xor}{\oplus}
\newcommand{\bigxor}{\bigoplus}

\title{Separating comb inequalities is NP-hard}
\author[1]{Yohan Finet\thanks{yohan.finet@usherbrooke.ca}}
\author[1]{Victor Drouin-Touchette\thanks{victor.drouin-touchette@usherbrooke.ca}}
\affil[1]{Institut Quantique and Département de Génie Électrique et de Génie Informatique, Faculté de Génie, Université de Sherbrooke, Sherbrooke, Québec, J1K 2R1, Canada}

\begin{document}
\maketitle

\begin{abstract}
Comb inequalities are important cutting planes for the symmetric travelling salesman problem, yet the complexity of their exact separation has remained a longstanding question in polyhedral combinatorics. We present a reduction from 3-SAT proving that deciding whether a comb inequality is violated is NP-complete and that the corresponding separation problem is NP-hard. This result holds even when the input vector belongs to the subtour elimination polytope, every edge value is zero, one half or one and the support graph is nonplanar with maximum degree four. The reduction constructs a graph in which six-vertex ladder gadgets encode Boolean relations and cubic graphs enforce consistency among occurrences of each logical variable. For a propositional logic formula with $v$ variables and $m$ clauses, the constructed graph has $40v+70m+30$ vertices and a linear number of positive edges. The reduction also proves hardness when every permitted tooth has two or four vertices with half of the tooth in the handle. We discuss consequences for approximating the maximum comb violation, optimization over the comb relaxation and explain why hardness of separation does not automatically transfer to larger inequality families.
\end{abstract}

\section{Introduction}\label{sec:introduction}

The symmetric travelling salesman problem (TSP) asks for a minimum-weight Hamiltonian cycle in an undirected weighted complete graph. A Hamiltonian cycle visits every vertex exactly once before returning to its starting vertex. Beyond its applications in logistics~\cite{Menger1931, Edmonds1965, Mahalanobis1940, Jessen1942}, biology~\cite{Cox1990, BenDor1997, BenDor2000}, electronics~\cite{Barbagallo1996, Gupta2003, Hirech1998, Makar1998, Lin1973, Grotschel1991} and astronomy~\cite{Carlson1997, ATCA2022}, the TSP has provided a setting in which many of the principal methods of combinatorial optimization were developed and tested~\cite{applegate2007}. Among these methods, cutting planes strengthen a linear programming (LP) relaxation by adding inequalities satisfied by every tour but violated by the current fractional solution~\cite{Dantzig1954}. Their integration with branching yields the branch-and-cut approach used in exact TSP solvers~\cite{Padberg1991, applegate2003}.

The effectiveness of this approach depends in part on the ability to find useful violated inequalities. This task is called \emph{separation}. An important class of such inequalities for the TSP is the class of comb inequalities (see details in Section~\ref{subsec:combs}). Introduced in their general form by Gr\"otschel and Padberg~\cite{Grotschel1979a,Grotschel1979b} in 1979 following earlier work of Edmonds~\cite{Edmonds1965} and Chv\'atal~\cite{Chvatal1973}, the computational complexity of their separation has remained unknown ever since~\cite{applegate2007, Cook2019}. Neither a polynomial-time exact separation algorithm nor an NP-hardness proof was known. It was however conjectured by Padberg and Gr\"otschel~\cite[p.~341]{padberg1985} that there exists a polynomial-time exact separation algorithm for this class of inequalities. Several positive results supporting this conjecture established tractability for restricted comb families or related classes of inequalities. A selection of these are discussed in Section~\ref{subsec:combs}.

In this paper, we prove NP-completeness of deciding if a comb inequality is violated by a reduction from 3-SAT. The reduction shows how Boolean relations, a ladder gadget and cubic graphs can be used to generate a fractional TSP solution from the propositional logic formula. The resulting fractional solution already satisfies every subtour elimination constraint and uses only the edge values $0$, $\tfrac12$, and $1$. An important difficulty is to ensure that an arbitrary violated comb encodes a satisfying assignment. Constructing a violated comb from a satisfying assignment is insufficient since without the reverse implication, unintended teeth could also produce violations for unsatisfiable formulas. We address this difficulty by classifying every tight cut in the constructed graph and showing that each tooth of a violated comb must lie inside a single ladder gadget.

The remainder of this paper is organized as follows. Sections~\ref{subsec:SEP} and~\ref{subsec:combs} introduce the notation and the separation problem. Section~\ref{sec:results} presents the main theorem and its corollaries. Section~\ref{sec:proof} presents the details of the proof with Section~\ref{subsec:equality} explaining how half-integrality constrains the teeth of violated combs and Section~\ref{subsec:ladder} describing the ladder gadget and the Boolean relation it encodes. Section~\ref{subsec:construction} shows how to construct a weighted graph from a 3-SAT formula while Section~\ref{subsec:correctness} proves that each tooth of a violated comb must lie inside a single ladder, shows how the constructed graph can be explicitly made nonplanar and proves both directions of the reduction. Finally, we present our conclusion in Section~\ref{sec:conclusion}.

\subsection{The subtour elimination polytope}\label{subsec:SEP}

Let $G=(V,E)$ be the complete undirected graph on $n=|V|$ vertices. For $S\subseteq V$, let $E(S)$ be the edges with both endpoints in $S$, and let
\begin{equation}
 \delta(S)=\{\{u,v\}\in E:u\in S,\ v\notin S\}
\end{equation}
be its cut. For disjoint vertex sets $A,B$, write $E(A,B)$ for the edges joining them. Given $x\in\mathbb{R}^{|E|}$ and $F\subseteq E$, put $x(F)=\sum_{e\in F}x_e$, and abbreviate $x(E(A,B))$ as $x(A,B)$.

The TSP can be formulated as an integer linear program (ILP) by introducing the edge weights $c_{ij}$ and the binary variables $x_{ij} \in \{0,1\}$ indicating whether edge $(i,j)$ is included in the tour \cite{Dantzig1954}:
\begin{subequations}\label{eq:TSP_ILP}
\begin{alignat}{3}
& \text{minimise} \quad
&& \sum_{i < j} c_{ij} x_{ij} \qquad
&& \label{eq:TSP_ILP:line1}\\
& \text{subject to} \quad
&& \sum_{j \ne i} x_{ij} = 2 \qquad
&& \forall i \in V \label{eq:TSP_ILP:line2}\\
&
&& x(E(S)) \leq |S| - 1 \qquad
&& \forall S \subset V \: : \: 2 \leq |S| \leq n - 1
   \label{eq:TSP_ILP:line3}\\
&
&& x_{ij} \geq 0 \qquad
&& \forall (i,j) \in E \label{eq:TSP_ILP:line4}\\
&
&& x_{ij} \in \{0, 1\} \qquad
&& \forall (i,j) \in E. \label{eq:TSP_ILP:line5}
\end{alignat}
\end{subequations}
Here, the degree constraints \eqref{eq:TSP_ILP:line2} ensure that every vertex is incident to exactly two tour edges, and the subtour elimination constraints \eqref{eq:TSP_ILP:line3} prevent the formation of disconnected cycles over subsets of~$V$. The convex hull in $\mathbb{R}^{|E|}$ of incidence vectors $x$ satisfying constraints \eqref{eq:TSP_ILP:line2} -- \eqref{eq:TSP_ILP:line5} is called the Symmetric Travelling Salesman Polytope and is denoted by $\STSP(n)$. The polytope obtained when removing \eqref{eq:TSP_ILP:line5} is called the Subtour Elimination Polytope and is denoted by $\SEP(n)$. For $n\geq3$, the upper bounds $x_e\leq1$ follow from Eq.~\eqref{eq:TSP_ILP:line3} when taking $S=e$. Summing the degree equations over $S$ also give the useful identity
\begin{equation}\label{eq:degree-identity}
 x(\delta(S)) + 2x(E(S)) = 2|S|.
\end{equation}
Consequently, the form \eqref{eq:TSP_ILP:line3} of the subtour elimination constraints is equivalent to
\begin{equation}\label{eq:SECs_restated}
    x(\delta(S)) \geq 2 \quad (\varnothing\ne S\subset V).
\end{equation}
Note, however, that the slack of Eq.~\eqref{eq:SECs_restated} is twice the slack of Eq.~\eqref{eq:TSP_ILP:line3}.

We use $x$ as a generic vector variable and $x^*$ for a fractional point. Given a class of inequalities $\mathcal{C}$ and a vector $x^*$, a separation algorithm for $\mathcal{C}$ either returns an inequality in $\mathcal{C}$ that is violated by $x^*$ or a failure message. A separation algorithm is exact if it returns a failure message only when all inequalities in $\mathcal{C}$ are satisfied by $x^*$, whereas a heuristic separation algorithm may return a failure message even when some inequality in $\mathcal{C}$ is violated by $x^*$.

$\SEP(n)$ is described by $n$ degree equations, $\binom{n}{2}$ nonnegativity constraints, and exponentially many subtour elimination constraints. Nevertheless, a linear objective can be optimized over $\SEP(n)$ in polynomial time. Indeed, the degree equations and nonnegativity constraints can be checked directly and a polynomial-time exact separation algorithm for subtour elimination constraints was shown by Hong~\cite{Hong1972} and subsequently improved by Crowder and Padberg~\cite{Crowder1980} and Padberg and Rinaldi~\cite{Padberg1990a}. The equivalence between separation and optimization~\cite[Theorem~6.4.9]{Grotschel1993} therefore implies polynomial-time optimization of a linear objective over $\SEP(n)$. The efficient separation of subtour elimination constraints makes $\SEP(n)$ a natural starting point for exact TSP methods. Solvers such as Concorde~\cite{applegate2003} combine subtour separation with routines for additional classes of valid inequalities to tighten their LP relaxation.

We call $x^*$ \emph{half-integral} when $2x^*$ is integral. On $\SEP(n)$ this means $x^*_e\in\{0,\tfrac12,1\}$. The \emph{support graph} is $G^*=(V,E^*)$, where $E^*=\{e:x^*_e>0\}$. Throughout the paper, a set $S$ is \emph{tight} if $\varnothing \neq S \subset V$ and has cut weight $x^*(\delta(S)) = 2$.

\subsection{Comb inequalities and exact separation}\label{subsec:combs}

Comb inequalities, as described by Gr\"otschel and Padberg~\cite{Grotschel1979a,Grotschel1979b}, are defined as follows. As illustrated in Fig.~\ref{fig:comb}, a comb consists of a handle $H\subset V$ and teeth $T_1,\ldots,T_s\subset V$ satisfying the following conditions:
\begin{subequations}\label{eq:comb-conditions}
\begin{align}
 &T_i\cap H\ne\varnothing &&(1\leq i\leq s),\\
 &T_i\setminus H\ne\varnothing &&(1\leq i\leq s),\\
 &T_i\cap T_j=\varnothing &&(i\ne j),\\
 &s\geq3 &&s\text{ odd}.
\end{align}
\end{subequations}
The comb inequality is
\begin{equation}\label{eq:comb}
 x(\delta(H))+\sum_{i=1}^{s}x(\delta(T_i))\geq3s+1
\end{equation}
and we denote its slack by
\begin{equation}\label{eq:slack}
 \sigma_x(H,T_1,\ldots,T_s)
 =x(\delta(H))+\sum_{i=1}^{s}x(\delta(T_i))-3s-1.
\end{equation}
A comb is violated when this quantity is negative.

\begin{figure}
  \centering
  \includegraphics[width=0.45\textwidth]{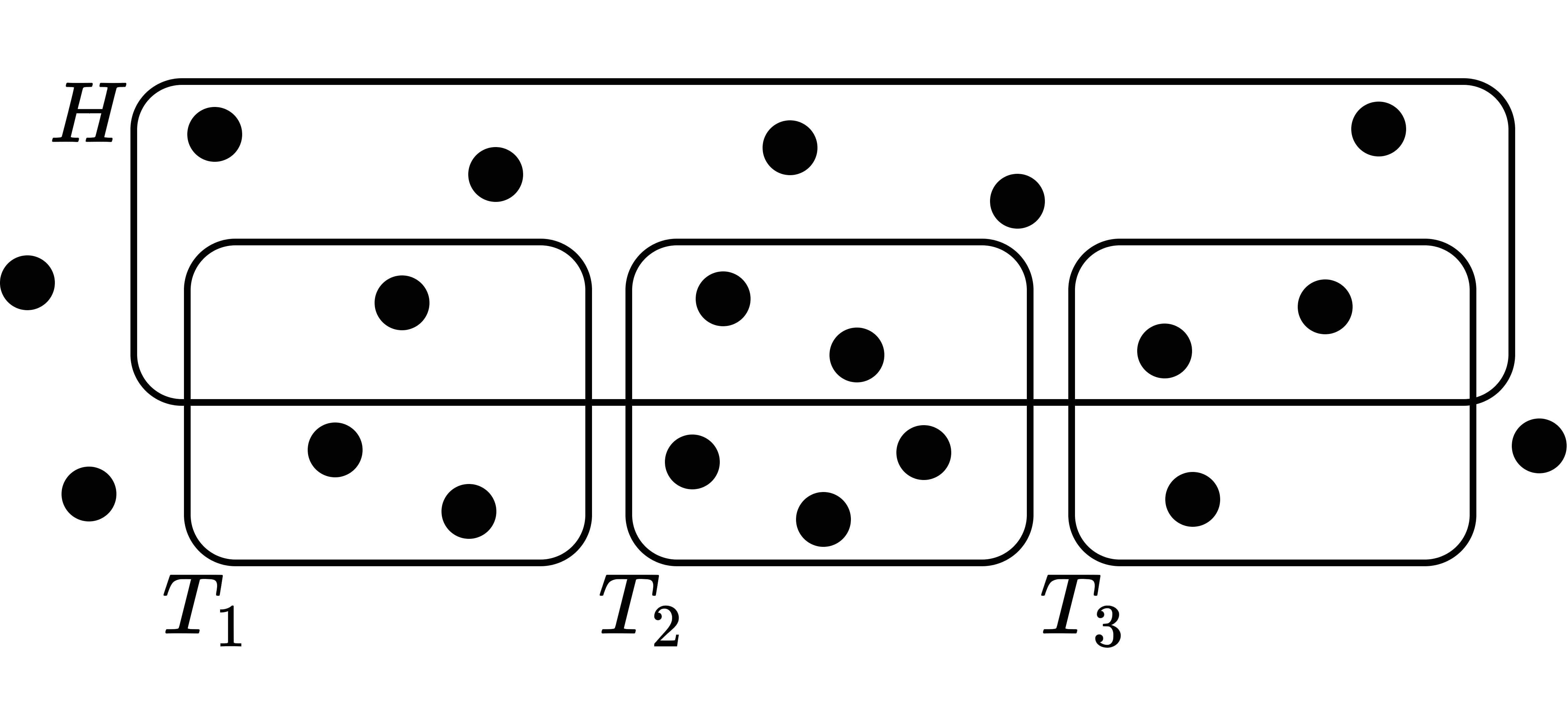}
  \caption{A comb with handle $H$ and teeth $T_1$, $T_2$ and $T_3$. The black dots represent vertices.}
  \label{fig:comb}
\end{figure}

Comb inequalities are powerful cutting planes for the symmetric travelling salesman problem because they exclude fractional solutions that satisfy all degree and subtour-elimination constraints, while preserving every feasible tour. Moreover, they were proven to induce facets of $\STSP(n)$ for ${n \geq 6}$~\cite{Grotschel1979a,Grotschel1979b}. Their practical importance was illustrated by Applegate et al.~\cite[Table~16.3]{applegate2007} who solved, using the Concorde solver~\cite{applegate2003}, 15 TSPLIB~\cite{reinelt1991} instances selected from those having between 1000 and 2500 cities. Disabling both the consecutive-ones and domino-based comb separation heuristics~\cite{applegate1995} more than doubled the reported running time on this test set, demonstrating the substantial contribution of the cuts provided by these routines to Concorde's computational performance. Additionally, for the much larger instance sw24978, they show that about 30\% of the non-subtour inequalities that were generated and used were comb inequalities. Similarly, for the very large pla85900 instance, about 31\% of the cutting planes defining the root LP relaxation (including subtour elimination constraints) were comb inequalities.

Polynomial-time separation algorithms are known for several restricted families of comb inequalities and related classes of inequalities. Padberg and Rao~\cite{Padberg1982} developed a polynomial-time separation algorithm for the 2-matching inequalities, also known as blossom inequalities, which correspond to comb inequalities whose teeth each contain two vertices. Carr~\cite{Carr1997} showed that comb inequalities with a fixed number of teeth $s$ can be separated by solving $O(n^{2s})$ maximum-flow problems, where $n$ is the number of vertices. Since $s$ may be any odd integer satisfying $3 \leq s \leq \lfloor n/2 \rfloor$, this approach does not yield a polynomial-time algorithm when the number of teeth is unrestricted. Fleischer and Tardos~\cite{Fleischer1999} described an $O(n^2 \log n)$-time algorithm that finds maximally violated combs when $x^* \in \SEP(n)$ and the support graph is planar. A comb is maximally violated if its slack is $-1$ in Eq.~\eqref{eq:slack}, the largest possible violation when $x^* \in \SEP(n)$. Letchford~\cite{Letchford2000} introduced the larger family of domino-parity inequalities and separated these under the same planarity assumption in $O(n^3)$ time. This result was further generalized by Cook, Espinoza and Goycoolea~\cite{Cook2010} to a larger family for which a subset containing all the clique-tree inequalities can be separated in polynomial time, provided that the support graph is planar and the number of handles is bounded by a fixed constant. Caprara, Fischetti and Letchford~\cite{Caprara2000} demonstrated that \{0,$\frac{1}{2}$\}-cuts form a family of inequalities that includes comb inequalities and showed how to detect maximally violated \{0,$\frac{1}{2}$\}-cuts in $O(n^2|E^*|)$ time. Fleischer, Letchford, and Lodi~\cite{Fleischer2006} also established polynomial-time separation for a superclass of simple comb inequalities when the fractional solution is in the subtour elimination polytope. Note that these algorithms for larger families may return violated inequalities that are not comb inequalities. Consequently, the results presented by Caprara, Fischetti and Letchford, for example, do not automatically extend to comb separation under the same assumptions since their algorithm may return a maximally violated \{0,$\frac{1}{2}$\}-cut that is not a comb inequality whether a maximally violated comb inequality exists or not. Other heuristic separation approaches were developed by Padberg and Rinaldi~\cite{Padberg1990b}, Clochard and Naddef~\cite{Clochard1993}, and Naddef and Thienel~\cite{Naddef2002}. Many other valid or faced-defining inequalities for $\STSP(n)$ exist, see for example \cite{Naddef2001} for a more extensive list.

\section{Results}\label{sec:results}

\subsection{NP-completeness and proof outline}

\begin{figure}[!b]
  \centering
  \includegraphics[width=1\textwidth]{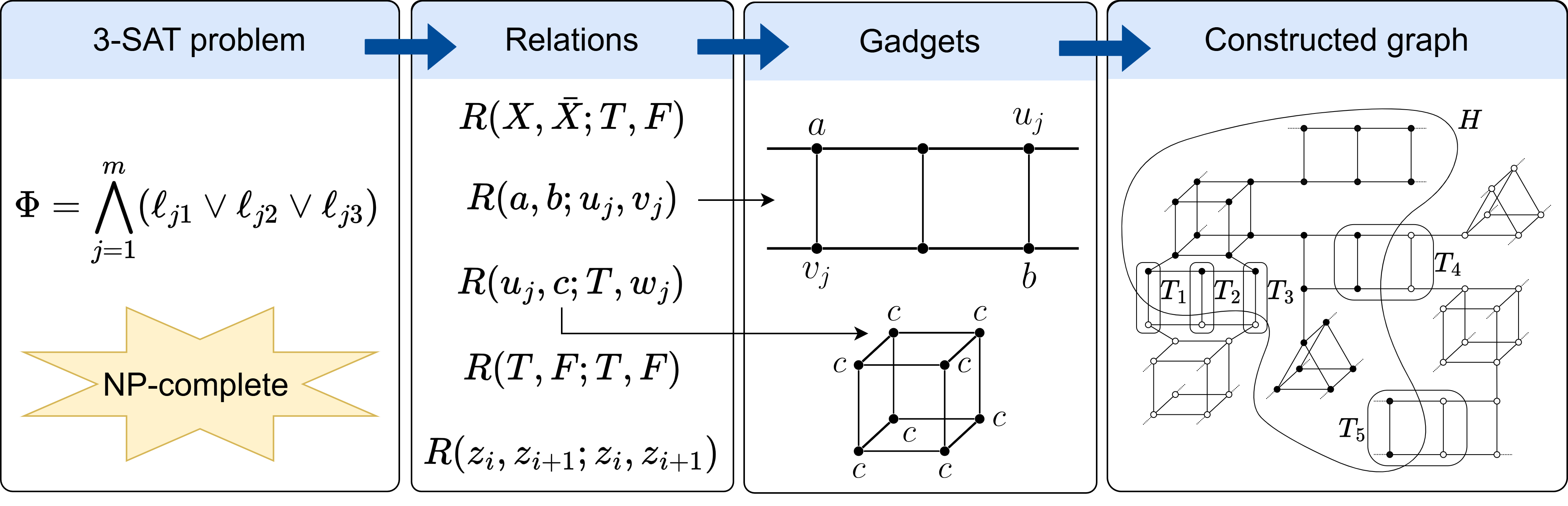}
  \caption{Outline of the reduction. From a 3-SAT formula, relations of the form $R(a,b;c,d): {a+b=c+d}$ are first introduced. Each relation is then encoded by a ladder gadget and each relation variable is represented by a cubic graph gadget whose vertices connect to the corresponding ladder ports. For a given comb, we color every vertex in the handle in black (assign value~1) and all other vertices in white (value~0). In any violated comb for this constructed graph, we show that ladder edges can only be cut in specific ways by the handle, only teeth of specific forms are allowed, each tooth must be contained in a ladder and each cubic graph gadget must lie entirely inside the handle or outside it. For a violated comb, the common vertex color of the cubic graph associated to each logical variable yields a satisfying assignment for the 3-SAT formula. Conversely, any satisfying assignment determines vertex colors and compatible handle and tooth choices that form a violated comb.}
  \label{fig:proof_outline}
\end{figure}

The decision problem corresponding to the separation of comb inequalities asks only whether a violated comb exists. Here the input is a fractional point of $\SEP(n)$, represented with finite binary encoding. This promise can be checked in polynomial time by verifying nonnegativity and the degree constraints and applying the separation procedure for subtour elimination constraints as discussed in Section~\ref{subsec:SEP}. Equivalently, one may define a decision language that rejects inputs outside $\SEP(n)$.

In this work, we use the standard complexity classes $\PP$ and $\NP$ and show a polynomial-time many-one reduction from 3-SAT to deciding if a comb is violated for an input $x^* \in \SEP(n)$. A many-one reduction constructs one target instance in polynomial time and preserves both yes and no answers. A problem is NP-hard if every problem in $\NP$ reduces to it in polynomial time, and NP-complete if it is also in $\NP$~\cite{GareyJohnson1979}.

NP-completeness of the decision problem for 3-SAT follows from the work of Cook~\cite{Cook71} and Karp~\cite{Karp1972}. The reduction to deciding comb violation, as illustrated in Fig~\ref{fig:proof_outline}, begins by introducing relations of the form $R(a,b;c,d): a+b=c+d$. These relations encode the constraints of a given 3-SAT formula and are used to show that each tooth of a violated comb must be contained within a certain subset. Each relation is represented by a ladder gadget, with each port assigned a variable of the relation, as shown for $R(a,b;u_j, v_j)$ with vertices $a$, $b$, $u_j$ and $v_j$ in Fig.~\ref{fig:proof_outline}. Each relation variable is encoded by a cubic graph (graph where each vertex has 3 edges) gadget. A graph is constructed by connecting the vertices of each cubic graph gadget to ladder ports corresponding to the same variable. Each edge of the graph is assigned a weight of $0$, $\tfrac12$ or $1$. We further show that the resulting constructed support graph can explicitly be made nonplanar and that the corresponding fractional point $x^* \in \SEP(n)$. As every vertex of the constructed graph has weighted degree two, as required by the degree equations, and each edge has a weight of $0$, $\tfrac12$ or $1$, the degree of a vertex in the support graph can be at most four.

We show that any violated comb in the constructed graph restricts how the handle can cut ladder edges and forces each cubic graph gadget to lie entirely inside or outside the handle. For these combs, we also show that each tooth must lie within a ladder. These restrictions result in having all vertices assigned to a given relation variable being either all inside the handle or all outside it. Coloring vertices inside the handle black (value~1) and those outside white (value~0) therefore assigns a Boolean value to each relation variable. The restrictions on the structure of violated combs ensure that these values satisfy the encoded relations and yield a satisfying assignment for the original 3-SAT formula.

Conversely, a satisfying assignment for the formula determines compatible vertex colors, which yields a handle and a choice of teeth that together form a violated comb. Thus, the formula is satisfiable if and only if the constructed instance admits a violated comb, which leads to Theorem~\ref{thm:main}.

\begin{theorem}\label{thm:main}
Deciding whether a fractional point $x^*\in\SEP(n)$ violates a comb inequality is \mbox{NP-complete}, even when every coordinate belongs to $\{0,\tfrac12,1\}$ and the support graph is nonplanar with maximum degree four. Consequently, the associated separation problem is NP-hard, even under these restrictions.
\end{theorem}

\begin{figure}[t]
  \centering
  \includegraphics[width=0.98\textwidth]{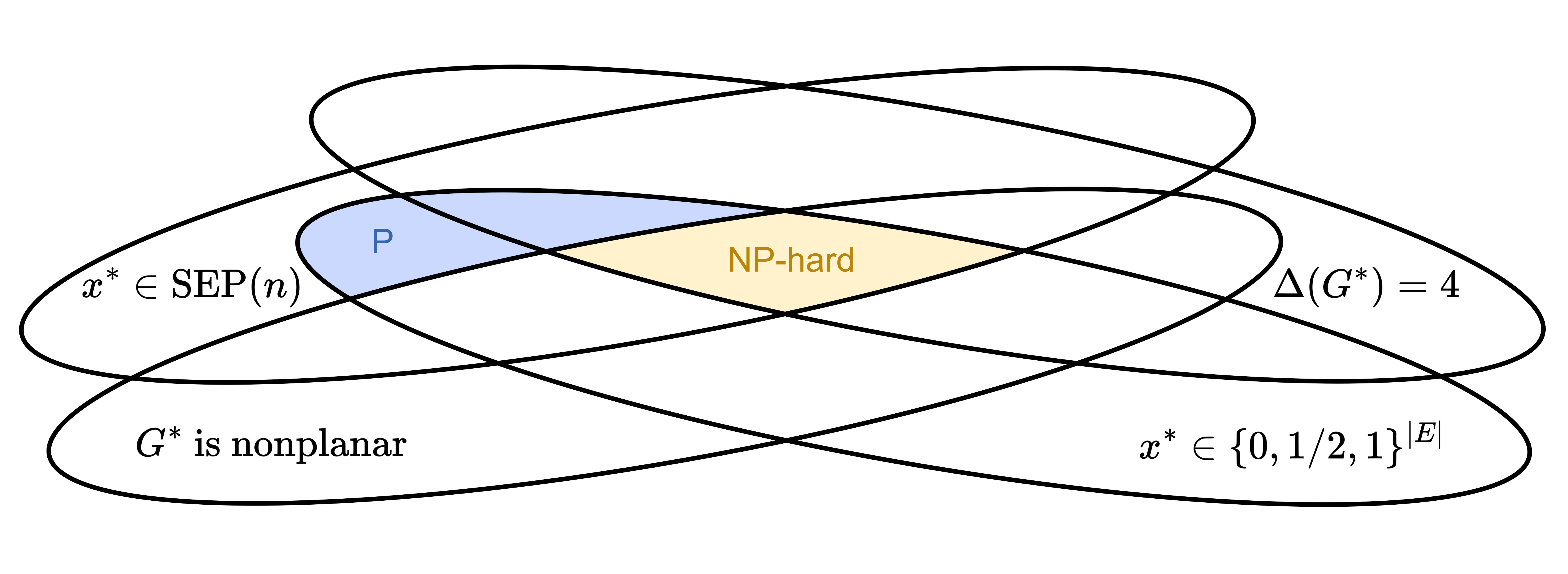}
  \caption{Venn diagram of the vectors considered in this paper. NP-hardness is proven for the vectors in the intersection of the four sets which establishes NP-hardness for arbitrary vectors. Some restricted classes nevertheless admit polynomial-time separation. In particular, by Lemma~\ref{lem:equality}, for $x^*\in\SEP(n)\cap\{0,\tfrac12,1\}^{|E|}$, every violated comb inequality is maximally violated. If the support graph $G^*$ is also planar, then the algorithm of Fleischer and Tardos~\cite{Fleischer1999} yields polynomial-time comb separation.}
  \label{fig:Venn_diagram}
\end{figure}

Our reduction proves NP-hardness for vectors in $\SEP(n)$, have each of their elements in $\{0,\tfrac12,1\}$ and have a nonplanar support graph with maximum degree $\Delta(G^*)=4$. Allowing arbitrary coordinates in $[0,1]^{|E|}$ retains NP-hardness because the constructed instances are still admissible inputs. However, this does not exclude polynomial-time separation for restricted classes of inputs as illustrated in Fig.~\ref{fig:Venn_diagram}. For instance, when $x^* \in \SEP(n)$ and $x^* \in \{0, \tfrac12, 1\}$, violated combs are all maximally violated (see Lemma~\ref{lem:equality}). Therefore, for these vectors, when $G^*$ is planar, the polynomial-time algorithm of Fleischer and Tardos~\cite{Fleischer1999} can be used so the separation problem for these inputs is in $\PP$. The hardness result neither establishes average-case hardness nor rules out polynomial-time separation for restricted comb families. For example, blossom inequalities admit polynomial-time separation~\cite{Padberg1982}.

\subsection{Corollaries and consequences}

This section presents several corollaries of Theorem~\ref{thm:main}. Short proofs are included here, while proofs that require intermediate results from the reduction are deferred to Section~\ref{subsubsec:cor_proofs}.

\begin{corollary}\label{cor:small-teeth}
Deciding whether a violated comb exists remains NP-complete if admissible teeth are restricted to size two or four, with $|T_i\cap H|=|T_i\setminus H|\in\{1,2\}$ for every tooth. The same restrictions on the input point as in Theorem~\ref{thm:main} apply.
\end{corollary}

This corollary distinguishes a bound on tooth size from a bound on the number of teeth. The latter grows with the formula in our construction, so there is no conflict with Carr's polynomial-time result for fixed numbers of handles and teeth~\cite{Carr1997}.

\begin{corollary}\label{cor:fixed-gap}
Under the slack definition of Eq.~\eqref{eq:slack}, define
\begin{equation}\label{eq:max-violation}
    V(x)=\max\bigl(\{0\}\cup\{-\sigma_{x}(C):C\text{ is a comb}\}\bigr).
\end{equation}
For every half-integral point $x^*\in\SEP(n)$, we have $V(x^*)\in\{0,1\}$. Unless $\PP=\NP$, no polynomial-time algorithm on such points can:
\begin{enumerate}
    \item always estimate $V(x^*)$ with absolute error strictly less than $1/2$; or
    \item for any fixed $\varepsilon\in(0,1]$, return a violated comb whenever a comb with violation at least $\varepsilon$ exists.
\end{enumerate}
In particular, the polynomial-time guarantee of Fleischer and Tardos~\cite{Fleischer1999} cannot extend to arbitrary support graphs unless $\PP=\NP$.
\end{corollary}

\begin{proof}
Lemma~\ref{lem:equality} states that every violated comb at a half-integral point of $\SEP(n)$ has a slack of exactly~$-1$. Thus $V(x^*)\in\{0,1\}$, with $V(x^*)=1$ precisely when $x^*$ violates a comb inequality.

Assume a polynomial-time algorithm with property 1 exists. An estimate $\widehat V(x^*)$ with absolute error strictly less than $1/2$ satisfies $\widehat V(x^*)>1/2$ if and only if $V(x^*)=1$, and therefore decides whether a violated comb exists. By Theorem~\ref{thm:main}, this decision problem is $\NP$-complete even on half-integral points of $\SEP(n)$. That would therefore imply $\PP=\NP$.

Likewise, assume a polynomial-time algorithm with property 2 exists and fix $\varepsilon\in(0,1]$. For half-integral points of $\SEP(n)$ every violated comb has violation $1\ge\varepsilon$, so the algorithm must return a violated comb whenever one exists. Checking any returned candidate in polynomial time therefore also decide whether a violated comb exists. By Theorem~\ref{thm:main} that would again imply $\PP=\NP$.
\end{proof}

Define the polytope $\CP(n)$ as the set of points that are in $\SEP(n)$ and satisfy every comb inequality:
\begin{equation}\label{eq:comb-relaxation}
    \CP(n)
    =\{x^*\in\SEP(n):\sigma_{x^*}(C)\geq0
      \text{ for every comb }C\}.
\end{equation}
$\coNP$ consists of the complements of languages in $\NP$~\cite{Arora2009}:
\begin{equation}
    \coNP = \{L : \overline{L} \in \NP\}.
\end{equation}

\begin{corollary}
\label{cor:comb-membership}
Given $n$ and a vector $x^*$ encoded in binary, deciding
whether $x^*\in \CP(n)$ is coNP-complete.
\end{corollary}

\begin{proof}
For a vector $x^*\notin\CP(n)$, a certificate consists of a violated nonnegativity constraint, degree equation, subtour elimination constraint or comb inequality. Each certificate has a polynomial-size description and can be checked in polynomial time. Thus, deciding if $x^*\notin\CP(n)$ is in $\NP$ and deciding if $x^*\in\CP(n)$ belongs to $\coNP$.

For $x^* \in \SEP(n)$, deciding if $x^*\notin\CP(n)$ is equivalent to deciding if a comb is violated. By Theorem~\ref{thm:main}, deciding whether a violated comb exists is NP-complete even when $x^* \in \SEP(x)$. Therefore, deciding if $x^*\notin\CP(n)$ is NP-hard and given the certificate argument above this proves that deciding if $x^*\notin\CP(n)$ is NP-complete.

The complement of an NP-complete language is coNP-complete. To see why, take an NP-complete language $L$. Since $L \in \NP$, we have $\overline{L} \in \text{coNP}$ by definition. For $A \in \coNP$, we have $\overline{A} \in \NP$. the NP-hardness of $L$ yields a polynomial-time reduction $f$ such that
\begin{equation}
    w \in \overline{A} \iff f(w) \in L.    
\end{equation}
By taking the complement of both statements we get
\begin{equation}
    w \in A \iff f(w) \in \overline{L},
\end{equation}
so $A \le_p \overline{L}$.
\end{proof}

A related consequence follows from the equivalence between separation and optimization for well-described polytopes~\cite{Grotschel1993}. In this setting, ``well-described'' means that a bound is known on the binary encoding length of each inequality, even though there may be exponentially many inequalities.

\begin{corollary}
\label{cor:comb-optimization}
Exact optimization of arbitrary linear objectives over
$\CP(n)$ is NP-hard under polynomial-time Turing
reductions.
\end{corollary}

\begin{proof}
The degree equations, variable bounds, and subtour elimination constraints all have small integral coefficients. For a comb inequality, each edge coefficient and right-hand side is bounded by $O(n)$. Since each inequality involves $O(n^2)$ variables, its binary encoding length is polynomial in $n$. Thus $\CP(n)$ is a well-described polyhedron.

By the optimization--separation equivalence~\cite{Grotschel1993}, an exact optimization oracle for $\CP(n)$ permits exact separation from $\CP(n)$, and hence membership testing for this polyhedron, in oracle-polynomial time. For $x^*\in\SEP(n)$, nonmembership in $\CP(n)$ is equivalent to the existence of a violated comb inequality. The oracle therefore solves the NP-complete decision problem of Theorem~\ref{thm:main}. Consequently, exact optimization over $\CP(n)$ is NP-hard under polynomial-time Turing reductions.
\end{proof}

It is tempting to infer hardness for any inequality family containing the comb inequalities. For separation, this inference requires care. Let $\cC$ be the family of comb inequalities and let $\cC \subseteq \mathcal{D}$, where $\mathcal{D}$ is a larger family. Family containment gives
\begin{equation}\label{eq:containment}
 \exists C\in\cC\text{ violated by }x^*
 \quad\Longrightarrow\quad
 \exists D\in\mathcal{D}\text{ violated by }x^*.
\end{equation}
The converse can fail; a point satisfying every comb inequality might violate an additional member of $\mathcal{D}$. In particular, a separator for $\mathcal{D}$ is allowed to return such a member without answering whether a violated comb inequality exists. A many-one reduction must preserve no instances as well as yes instances.

For the present construction to establish hardness of $\mathcal{D}$-separation by the same map, one would additionally need to prove for a 3-SAT formula $\Phi$ that
\begin{equation}\label{eq:transfer-condition}
 \Phi\text{ unsatisfiable}
 \quad\Longrightarrow\quad
 x^*(\Phi)\text{ satisfies every inequality in }\mathcal{D}.
\end{equation}
Together with containment and the reduction presented herein, this condition would give the required equivalence. We have proved it for combs, but make no such claim for larger families such as domino-parity, clique-tree and path inequalities.

\section{Proof details}\label{sec:proof}

\subsection{Structure of violated combs}\label{subsec:equality}

Before constructing the reduction, we describe the structure of any violated comb at a half-integral subtour-feasible point. The slack calculations presented here are closely related to those presented in the work of Letchford~\cite{Letchford2000}. 

\begin{definition}\label{def:domino}
Following the definition of Letchford~\cite{Letchford2000}, a \emph{domino} is a pair $(A,B)$ with the following properties:
\begin{itemize}
    \item $\varnothing \neq A,B \subset V$,
    \item $A \cap B = \varnothing$,
    \item $A \cup B \neq V$.
\end{itemize}
\end{definition}

As noted earlier, the slack of Eq.~\eqref{eq:SECs_restated} is twice the slack of Eq.~\eqref{eq:TSP_ILP:line3}. Therefore, for a nonempty proper subset $S$, denote the slack of its subtour elimination constraint by
\begin{equation}\label{eq:excess}
 e_x(S)=\frac{x(\delta(S))-2}{2}.
\end{equation}
For a domino $(A,B)$, define its weight as
\begin{equation}\label{eq:domino-weight}
 w(A,B)=e_x(A)+e_x(B)+e_x(A\cup B).
\end{equation}
Recall that $x(A,B)=x(E(A,B))$. For a comb, let $A_i=T_i\cap H$, $B_i=T_i\setminus H$, and define the residual as the difference between the weight of the handle cut and the total weight of edges joining the two parts of each tooth:
\begin{equation}\label{eq:residual}
 r_x=x(\delta(H))-\sum_{i=1}^{s}x(A_i,B_i).
\end{equation}
Furthermore, note that
\begin{equation}\label{eq:cut-sum}
 x(\delta(A))+x(\delta(B))
 =x(\delta(A\cup B))+2x(A,B).
\end{equation}
Substituting Eq.~\eqref{eq:excess} in Eq.~\eqref{eq:domino-weight} and then using Eq.~\eqref{eq:cut-sum} gives $w(A,B)=x(\delta(A\cup B))+x(A,B)-3$. Rewriting the tooth terms in Eq.~\eqref{eq:slack} using this identity and substituting the expression for $x(\delta(H))$ from Eq.~\eqref{eq:residual} gives
\begin{equation}\label{eq:slack-identity}
    \sigma_x=r_x+\sum_{i=1}^{s}w(A_i,B_i)-1.
\end{equation}

\begin{lemma}\label{lem:equality}
If $x^*\in\SEP(n)$ is half-integral and a comb is violated, then its slack is $-1$, its residual is zero, and $A_i$, $B_i$, and $T_i$ are tight for every tooth.
\end{lemma}
\begin{proof}
If $x^*\in\SEP(n)$, then all slacks in Eq.~\eqref{eq:domino-weight} are nonnegative. Moreover, the edge sets $E(A_i,B_i)$ are disjoint subsets of $\delta(H)$, so nonnegativity constraints imply $r_{x^*}\geq0$. Eq.~\eqref{eq:slack-identity} therefore gives $\sigma_{x^*}\geq-1$. By Eq.~\eqref{eq:degree-identity}, every cut weight is integral when $x^*$ is half-integral. Consequently, Eq.~\eqref{eq:slack} implies that $\sigma_{x^*}$ is an integer. Thus, a violated comb has a negative slack that must be equal to $-1$. When $\sigma_{x^*}$ is negative, equality in Eq.~\eqref{eq:slack-identity} forces $r_{x^*}=0$ and every $w(A_i,B_i)=0$. Hence, each of the three nonnegative slacks defining $w$ must vanish. In other words, $A_i$, $B_i$, and $T_i$ are tight for every $i$.
\end{proof}

\begin{definition}\label{def:tight-domino}
A \emph{tight domino}~\cite{applegate2007} is a domino for which $x(\delta(A)) = x(\delta(B)) = x(\delta(A \cup B)) = 2$. Its \emph{semicut} consists of the positive edges in $E(A,B)$~\cite{Fleischer1999}.
\end{definition}

For the considered points, Lemma~\ref{lem:equality} says that the two halves of every tooth of a violated comb form a tight domino. Consequently, Eq.~\eqref{eq:cut-sum} shows that every semicut in a violated comb has weight one. Lemma~\ref{lem:equality} tells more than the tightness of individual teeth for violated combs: the handle-cut edges must be exactly the union of the semicuts. Indeed, $r_{x^*}$ is the total weight of all positive handle-cut edges outside that union and it must be equal to zero. Conversely, an odd collection of at least three disjoint tight dominos, equipped with a handle whose cut is the union of their semicuts, has slack $-1$ by Eq.~\eqref{eq:slack-identity}.

\subsection{A six-vertex ladder encoding a Boolean relation}\label{subsec:ladder}

\subsubsection{The ladder and its tight sets}

We now introduce the weighted ladder graph illustrated in Fig.~\ref{fig:ladder}, with vertex set
\begin{equation}
 U=\{a_1,a_2,a_3,b_1,b_2,b_3\}.
\end{equation}
This will be one of the main building blocks of the reduction and will be used to encode Boolean relations of the form $R(a,b;c,d):a+b=c+d$. These relations will be derived from the 3-SAT formula. The three edges $a_i b_i$ have weight one and are called rungs. The four edges $a_1a_2$, $a_2a_3$, $b_1b_2$, and $b_2b_3$ have weight one half and are called rails. Each corner $a_1,a_3,b_1,b_3$ receives one external edge of weight one half. These corners are the four \emph{ports}. All six vertices then have weighted degree two, and $x^*(\delta(U))=2$.

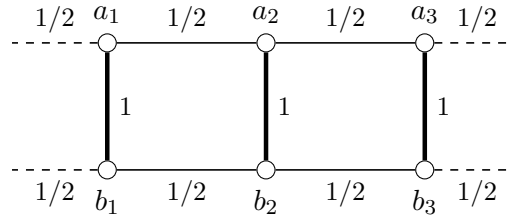
\begin{figure}[htbp]
\centering
\begin{tikzpicture}[scale=1.05,
 vertex/.style={circle,draw,fill=white,inner sep=2.4pt},
 half/.style={line width=.6pt},
 unit/.style={line width=1.7pt}]
 \foreach \i/\xx in {1/0,2/2,3/4}{
  \node[vertex,label=above:$a_\i$] (a\i) at (\xx,1.6) {};
  \node[vertex,label=below:$b_\i$] (b\i) at (\xx,0) {};
  \draw[unit] (a\i)--node[right,font=\small] {$1$} (b\i);
 }
 \foreach \i/\j in {1/2,2/3}{
  \draw[half] (a\i)--node[above,font=\small] {$1/2$} (a\j);
  \draw[half] (b\i)--node[below,font=\small] {$1/2$} (b\j);
 }
 \draw[half,dashed] (a1)--node[above,font=\small] {$1/2$}
  (-1.2,1.6);
 \draw[half,dashed] (b1)--node[below,font=\small] {$1/2$}
  (-1.2,0);
 \draw[half,dashed] (a3)--node[above,font=\small] {$1/2$}
  (5.2,1.6);
 \draw[half,dashed] (b3)--node[below,font=\small] {$1/2$}
  (5.2,0);
\end{tikzpicture}
\caption{The six-vertex ladder. Each rung $a_i b_i$ has weight one, and each rail edge has weight $1/2$. The four ports $a_1,a_3,b_1,b_3$ each have an external edge of weight $1/2$, shown dashed. In the rooted version, the four external endpoints are identified as a single vertex.}
\label{fig:ladder}
\end{figure}

We turn to a study of local cuts for a moment. Temporarily attach all four external edges to one root $r$. Let $L^\circ$ denote this seven-vertex weighted graph. For $S\subseteq U$, let $k(S)$ count the split rungs, meaning rungs with exactly one endpoint in $S$. Let $\rho_a(S)$ and $\rho_b(S)$ count the occupied intervals on the two three-vertex rails. An occupied interval is a maximal consecutive nonempty sequence of vertices belonging to $S$ along that rail. For example, $\rho_a(S) = 1$ for $S=\{a_1\}$, $S=\{a_1, a_2\}$ and $S=\{a_1, a_2, a_3\}$ while $\rho_a(S) = 2$ for $S=\{a_1, a_3\}$.

\begin{lemma}\label{lem:ladder-tight}
The tight sets contained in $U$ are precisely the six singletons and the six sets $P_\ell\cup\cdots\cup P_j$ with $1\leq\ell\leq j\leq3$ where $P_i=\{a_i,b_i\}$. The vector $x^*$ represented by the rooted graph $L^\circ$ belongs to $\SEP(7)$.
\end{lemma}
\begin{proof}
First note that for every $S\subseteq U$,
\begin{equation}\label{eq:ladder-cut}
 x^*(\delta(S))=k(S)+\rho_a(S)+\rho_b(S).
\end{equation}
Indeed, each occupied rail interval contributes two boundary half-weight edges, counting an edge to the root when the interval meets an end. Its contribution is therefore one. Each split rung contributes one, proving Eq.~\eqref{eq:ladder-cut}.

For nonempty $S$, if $k(S)=0$, then $S$ is a union of whole rungs. Both rails are occupied, and the cut weight is at least two. We have $x^*(\delta(S))=2$ when the selected rungs are consecutive. Thus, the sets $P_\ell\cup\cdots\cup P_j$ with $1\leq\ell\leq j\leq3$ are tight. If $k(S)\geq1$, at least one rail interval is present, again giving cut weight at least two. Having $x^*(\delta(S))=2$ requires $k=1$ and $\rho_a+\rho_b=1$. In this case, only one rail is occupied, and only one rung is split, so $S$ is a singleton.

Showing that the vector represented by $L^\circ$ belongs to $\SEP(7)$ requires showing it satisfies subtour elimination constraints and degree constraints. Nonempty proper subsets containing $r$ have the same cut weight as their complement which is always greater or equal to two as we just showed. Finally, the root has weighted degree two, as do all six other vertices.
\end{proof}

The tight-set classification immediately gives the complete list of internal tight dominos:
\begin{equation}\label{eq:local-teeth}
\begin{split}
 &(\{a_i\},\{b_i\})\quad(i=1,2,3),\\
 &(P_1,P_2),\qquad (P_2,P_3),\\
 &(P_1,(P_2\cup P_3)),\qquad((P_1\cup P_2),P_3).
\end{split}
\end{equation}
To see that the list is exhaustive, a tight singleton can join another disjoint tight set to form a tight union only when the union is its whole rung. If both halves are unions of whole rungs, they must be adjacent.

\subsubsection{The boundary relation and tooth parity}

Represent a handle $H$ by the function $h:V\to\{0,1\}$, where $h(v)=1$ if $v\in H$ and $h(v)=0$ otherwise. We refer to these two values as the \emph{colors} of the vertices. An edge belongs to $\delta(H)$ exactly when its endpoints have different colors. Consider a ladder with a collection of disjoint internal tight dominos. Let its internal handle cut be the union of their semicuts and assume that every tooth is split according to $h$. See Fig.~\ref{fig:ladder-configurations} (b) and (c) for examples.

\begin{lemma}\label{lem:relation}
Such a configuration exists for prescribed port colors if and only if the boundary condition
\begin{equation}\label{eq:boundary}
 h(a_1)+h(b_3)=h(a_3)+h(b_1)
\end{equation}
is satisfied. In every such configuration, the number $s_U$ of teeth satisfies
\begin{equation}\label{eq:parity}
 s_U\equiv h(a_1)+h(b_3)\pmod2.
\end{equation}
Every feasible port assignment has a realization using only teeth of size two or four.
\end{lemma}
\begin{proof}
First suppose the ladder contains only rung teeth, which are tight dominos of the form $T=(\{a_i\},\{b_i\})$ (see Fig.~\ref{fig:ladder-configurations} (b)). No rail edge belongs to their semicuts, so both rails must be uncut. Each rail consequently has a constant color. If the rail colors agree, then all vertices have the same color and there's no internal tooth being cut by the handle. If the rail colors differ, all three rungs are cut and all three rung teeth must intersect with the handle.

If the ladder does not contain only rung teeth, at least one tooth uses at least two full rungs\footnote{The two halves of every internal tooth considered in Lemma~\ref{lem:relation} form a tight domino. The complete list in Eq.~\eqref{eq:local-teeth} consists of rung teeth and teeth containing at least two full rungs. Thus, if the teeth in the ladder are not all rung teeth, one of them must contain at least two full rungs.} (see Fig.~\ref{fig:ladder-configurations} (c)). The list in Eq.~\eqref{eq:local-teeth} shows that its semicut consists of the two rail edges at one boundary between consecutive rungs. There can be at most one such tooth because any two overlap. The rungs within it are uncut, and the rail edges at the other boundary are uncut. If the tooth leaves one rung available, those uncut rails force the two endpoints of that rung to have the same color. Thus there is exactly one tooth, all rungs are uncut, and the two end rungs have opposite colors.

These cases produce the six port colorings in Table~\ref{tab:states}, three of which are presented in Fig.~\ref{fig:ladder-configurations}. They are exactly the Boolean solutions of Eq.~\eqref{eq:boundary}, and their associated internal tooth counts have the parity in Eq.~\eqref{eq:parity}. Uniform colors realize the two uniform vertex assignments to the handle, opposite rail colors realize two assignments using three rung teeth, and opposite end-rung colors realize the last two assignments. In this last case, only the four-vertex tooth $P_1\mid P_2$ is needed if we give rungs two and three the same color. The other four-vertex and six-vertex teeth could also be used without changing the parity of Eq.~\eqref{eq:parity}.
\end{proof}

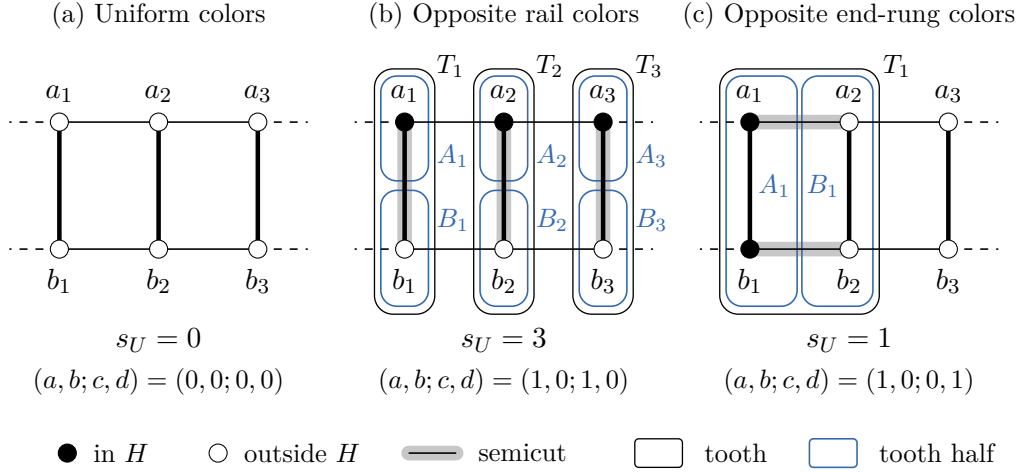
\begin{figure}[t]
\centering
\definecolor{toothhalfcolor}{HTML}{3568AD}
\def\semicutwidth{5.5pt}

\begin{tikzpicture}[scale=1.05,
 vertex/.style={circle,draw,fill=white,inner sep=2.4pt},
 half/.style={line width=.6pt},
 unit/.style={line width=1.7pt},
 tooth/.style={draw=black,line width=.5pt,rounded corners=7pt},
 toothhalf/.style={draw=toothhalfcolor,line width=.6pt,
                   rounded corners=6pt},
 semicut/.style={draw=black!22,line width=\semicutwidth,
                 line cap=round}]

 \foreach \panel/\shift/\title in
  {0/0/{(a) Uniform colors},
   1/4.35/{(b) Opposite rail colors},
   2/8.70/{(c) Opposite end-rung colors}}{

  \begin{scope}[shift={(\shift,0)}]

   \foreach \i/\xx in {1/0,2/1.25,3/2.50}{
    \coordinate (a\i) at (\xx,1.6);
    \coordinate (b\i) at (\xx,0);
   }

   \ifnum\panel=1
    \foreach \i in {1,2,3}
     \draw[semicut] (a\i)--(b\i);
   \fi

   \ifnum\panel=2
    \draw[semicut] (a1)--(a2);
    \draw[semicut] (b1)--(b2);
   \fi

   \ifnum\panel=1
    \foreach \i/\xx in {1/0,2/1.25,3/2.50}{
     \draw[tooth] ({\xx-.38},-.82)
      rectangle ({\xx+.38},2.27);

     \draw[toothhalf] ({\xx-.30},.86)
      rectangle ({\xx+.30},2.18);

     \draw[toothhalf] ({\xx-.30},-.72)
      rectangle ({\xx+.30},.74);
    }
   \fi

   \ifnum\panel=2
    \draw[tooth] (-.38,-.82) rectangle (1.63,2.27);
    \draw[toothhalf] (-.30,-.72) rectangle (.595,2.18);
    \draw[toothhalf] (.655,-.72) rectangle (1.55,2.18);
   \fi

   \foreach \i in {1,2,3}
    \draw[unit] (a\i)--(b\i);

   \foreach \i/\j in {1/2,2/3}{
    \draw[half] (a\i)--(a\j);
    \draw[half] (b\i)--(b\j);
   }

   \draw[half,dashed] (a1)--(-.63,1.6);
   \draw[half,dashed] (b1)--(-.63,0);
   \draw[half,dashed] (a3)--(3.13,1.6);
   \draw[half,dashed] (b3)--(3.13,0);

   \foreach \i in {1,2,3}{
    \ifnum\panel=0
     \node[vertex,label=above:$a_\i$] at (a\i) {};
     \node[vertex,label=below:$b_\i$] at (b\i) {};
    \fi

    \ifnum\panel=1
     \node[vertex,fill=black,label=above:$a_\i$]
      at (a\i) {};
     \node[vertex,label=below:$b_\i$] at (b\i) {};
    \fi

    \ifnum\panel=2
     \ifnum\i=1
      \node[vertex,fill=black,label=above:$a_\i$]
       at (a\i) {};
      \node[vertex,fill=black,label=below:$b_\i$]
       at (b\i) {};
     \else
      \node[vertex,label=above:$a_\i$] at (a\i) {};
      \node[vertex,label=below:$b_\i$] at (b\i) {};
     \fi
    \fi
   }

    \ifnum\panel=1
    \foreach \i/\xx in {1/0,2/1.25,3/2.50}{
    \node[anchor=west,inner sep=0pt,font=\small]
    at ({\xx+.40},2.31) {$T_\i$};

    \node[anchor=west,inner sep=0pt,font=\small,
            text=toothhalfcolor]
    at ({\xx+.40},1.17) {$A_\i$};

    \node[anchor=west,inner sep=0pt,font=\small,
            text=toothhalfcolor]
    at ({\xx+.40},.37) {$B_\i$};
    }
    \fi

   \ifnum\panel=2
    \node[anchor=west,font=\small] at (1.55,2.31) {$T_1$};
    \node[font=\small,text=toothhalfcolor]
     at (.31,.80) {$A_1$};
    \node[font=\small,text=toothhalfcolor]
     at (.92,.80) {$B_1$};
   \fi

   \node[font=\small] at (1.25,2.95) {\title};

   \ifnum\panel=0
    \node at (1.25,-1.15) {$s_U=0$};
    \node[font=\small] at (1.25,-1.67)
     {$(a,b;c,d)=(0,0;0,0)$};
   \fi

   \ifnum\panel=1
    \node at (1.25,-1.15) {$s_U=3$};
    \node[font=\small] at (1.25,-1.67)
     {$(a,b;c,d)=(1,0;1,0)$};
   \fi

   \ifnum\panel=2
    \node at (1.25,-1.15) {$s_U=1$};
    \node[font=\small] at (1.25,-1.67)
     {$(a,b;c,d)=(1,0;0,1)$};
   \fi

  \end{scope}
 }

 \node[vertex,fill=black] at (.10,-2.57) {};
 \node[anchor=west,font=\small] at (.30,-2.57)
  {in $H$};

 \node[vertex] at (2.00,-2.57) {};
 \node[anchor=west,font=\small] at (2.20,-2.57)
  {outside $H$};

 \draw[semicut] (4.40,-2.57)--(5.00,-2.57);
 \draw[half] (4.40,-2.57)--(5.00,-2.57);
 \node[anchor=west,font=\small] at (5.15,-2.57)
  {semicut};

 \draw[tooth,rounded corners=1.5pt]
  (7.25,-2.75) rectangle (7.85,-2.39);
 \node[anchor=west,font=\small] at (8.00,-2.57)
  {tooth};

 \draw[toothhalf,rounded corners=1.5pt]
  (9.45,-2.75) rectangle (10.05,-2.39);
 \node[anchor=west,font=\small] at (10.20,-2.57)
  {tooth half};

\end{tikzpicture}
\caption{Representative ladder configurations realizing the six feasible port colorings. Black vertices belong to the handle and white vertices do not. Complementing all colors gives the other three realizations. The port order is $(a,b,c,d)=(h(a_1),h(b_3),h(a_3),h(b_1))$. Rounded outlines delimit teeth and their halves while grey bands mark their semicuts. Thick rungs have weight one while thin rails and dashed external edges have weight $1/2$. The internal handle cut is exactly the union of the semicuts. In every panel, $a+b=c+d$ and $s_U\equiv a+b\pmod 2$.}
\label{fig:ladder-configurations}
\end{figure}

\begin{table}[t]
\centering
\begin{tabular}{cccccl}
\toprule
$a$ & $b$ & $c$ & $d$ & $s_U$ & Chosen realization\\
\midrule
0&0&0&0&0&All vertices outside the handle\\
1&1&1&1&0&All vertices inside the handle\\
1&0&1&0&3&Top rail inside, bottom rail outside\\
0&1&0&1&3&Top rail outside, bottom rail inside\\
1&0&0&1&1&Rung one inside, rungs two and three outside\\
0&1&1&0&1&Rung one outside, rungs two and three inside\\
\bottomrule
\end{tabular}
\caption{All feasible assignments to $(a,b,c,d)=(h(a_1),h(b_3),h(a_3),h(b_1))$ and the respective number of teeth.}
\label{tab:states}
\end{table}

To express the ladder constraint as a Boolean relation, write $a=h(a_1)$, $b=h(b_3)$, $c=h(a_3)$, and $d=h(b_1)$. These symbols denote the port colors. The boundary condition then becomes $a+b=c+d$, which we denote by $R(a,b;c,d)$.

Addition in this relation is ordinary integer addition. By Lemma~\ref{lem:relation}, its associated tooth parity is $a\xor b$, where $\xor$ denotes addition modulo two. Eq.~\eqref{eq:boundary} and Eq.~\eqref{eq:parity} are complementary: the relation excludes $(1,1;0,0)$ and $(0,0;1,1)$ even though both satisfy the corresponding equation modulo two.

\subsubsection{The rooted ladder}

The following fact will later exclude teeth that extend across several gadgets.

\begin{lemma}\label{lem:rooted-free}
The rooted graph $L^\circ$ has no violated comb inequality.
\end{lemma}
\begin{proof}
Every tooth contains at least two vertices. A comb on seven vertices therefore has exactly three teeth, with sizes $(2,2,2)$ or $(2,2,3)$. If violated, the two halves of each tooth form a tight domino by Lemma~\ref{lem:equality}. A two-vertex tight domino requires its two vertices to be joined by a unit edge, by Eq.~\eqref{eq:cut-sum}. The only such edges are the three rungs.

In the first case, the teeth are the three rungs. Removing their semicuts leaves both rails joined through the root, so the remaining support graph is connected. Therefore, every nonempty proper handle $H$ cuts at least one remaining edge. Since every edge in the support graph has positive weight, this edge contributes positively to the residual, contradicting the zero-residual condition (Lemma~\ref{lem:equality}).

In the second case, two teeth are rungs and the last tooth consists of the root and the remaining rung. This triple is tight only when the remaining rung is an end rung: its complement consists of two consecutive whole rungs, as classified by Lemma~\ref{lem:ladder-tight}. The only tight split of the triple is the root versus that end rung. A pair consisting of the root and one endpoint has cut weight three and is not tight. Its semicut is thus the pair of external edges at that end. After deleting this semicut and the other two rung semicuts, the rails and the uncut end rung connect all six ladder vertices, and the root is still attached at the opposite end. The remaining support is again connected, giving the same contradiction.
\end{proof}

\subsection{From 3-SAT to a weighted graph}\label{subsec:construction}

The reduction proceeds in the following stages. We first encode the input 3-SAT formula using Boolean relations of the form $R(a,b;c,d)$. This relation is encoded using the ladder gadget introduced in Section~\ref{subsec:ladder}, with its four arguments corresponding to the colors of the four ports. Lemma~\ref{lem:relation} provides the connection: a port coloring admits the required local tooth configuration exactly when its four values satisfy $R$, and the parity of the number of selected teeth is determined by these values.

As such, Boolean values of this formula will be encoded through  handle membership in the resulting graph. Color one means that a vertex belongs to $H$, while color zero means that it lies outside $H$. Occurrences of the same logical variable will be connected through a variable gadget (a \textit{prism} graph) designed to enforce a common color for these occurrences in every violated comb. This consistency property, together with the necessary restrictions on tooth structure, will be proved in Section~\ref{subsec:correctness}. This realizes a correspondence between Boolean assignments satisfying the 3-SAT formula and handles equipped with compatible teeth of violated combs. The graph and its edge weights are constructed from the formula alone. No satisfying assignment is needed to build them.

\subsubsection{Encoding clauses}\label{subsubsec:encoding}

The source problem is 3-SAT, whose NP-completeness follows from the work of Cook~\cite{Cook71} and Karp~\cite{Karp1972}. Its input is a Boolean formula
\begin{equation}
 \Phi=\bigwedge_{j=1}^{m}(\ell_{j1}\lor\ell_{j2}\lor\ell_{j3}),
\end{equation}
where each literal $\ell_{jk}$ is a source variable or its negation. A satisfying assignment makes at least one literal true in each clause. We assume $v\geq1$ source variables, $m\geq1$ clauses and exactly three literals per clause.

Consider two logical reference variables $T$ and $F$, intended to represent truth and falsehood, respectively. In the graph constructed by the reduction, the value of these logical variables will be determined by the handle colors of their representing vertices, just as for the other logical variables. We will enforce $T\ne F$ through the constraint that combs must be composed of an odd number of teeth; this will be covered in Section~\ref{subsec:anchor_relation}. Once their values differ, complementing the handle if necessary allows us to take $T=1$ and $F=0$ (see Section~\ref{subsec:anchor_relation}). In this subsection, we describe the encoding of clauses. For each source variable of a 3-SAT formula, we introduce the logical variables $X$ and $\bar X$, and enforce consistency between these variables using a \textit{complement relation}:
\begin{equation}\label{eq:complement}
 R(X,\bar X;T,F).
\end{equation}
When $T=1$ and $F=0$, this relation becomes $X+\bar X=1$.

For a clause $j$ whose literal representatives are $a,b,c$, introduce three new logical variables $u_j,v_j,w_j$ and impose the \textit{clause relations}:
\begin{equation}\label{eq:clause-relations}
 R(a,b;u_j,v_j),\qquad R(u_j,c;T,w_j).
\end{equation}
We use $X$ to represent a positive literal and $\bar X$ to represent a negative one.

\begin{lemma}\label{lem:clause}
For $T=1$ and any $a,b,c\in\{0,1\}$, there exist ${u_j,v_j,w_j\in\{0,1\}}$ satisfying both relations in Eq.~\eqref{eq:clause-relations} if and only if $a\lor b\lor c$ is true.
\end{lemma}
\begin{proof}
The two equations are $a+b=u_j+v_j$ and $u_j+c=1+w_j$. First suppose these equations are satisfied. Substituting $u_j$ in the second equation gives
\begin{equation}\label{eq:clause-sum}
 a+b+c=1+v_j+w_j\geq1,
\end{equation}
so the clause is satisfied. Conversely, if the clause is true, choose
\begin{equation}\label{eq:clause-witness}
 u_j=a\lor b,\qquad v_j=a\land b,\qquad w_j=u_j+c-1.
\end{equation}
For Boolean $a,b$, their sum equals their disjunction plus their conjunction: ${a+b=(a \lor b) + (a \land b)}$. With the choice of $u_j$ and $v_j$ from Eq.~\eqref{eq:clause-witness}, this proves that the first relation is satisfied. Next, rewrite $w_j$ as $w_j=(u_j\lor c) + (u_j\land c) - 1$. Clause satisfaction gives $u_j\lor c=1$, so $w_j\in\{0,1\}$. Substituting the choice of $w_j$ from Eq.~\eqref{eq:clause-witness} in the second relation shows it will be satisfied.
\end{proof}

\subsubsection{The anchor relation}\label{subsec:anchor_relation}

We call the logical reference variables $T$ and $F$ the \emph{anchors}, because they provide a common reference for interpreting truth and falsehood throughout the construction. Their names do not prescribe their handle colors. Instead, we must ensure that they receive opposite colors. In this subsection, we show that we can complement the handle if needed to obtain $T=1$ and $F=0$.

To enforce $T\ne F$, take two copies of every relation in Eqs.~\eqref{eq:complement} and~\eqref{eq:clause-relations}, and add one copy of
\begin{equation}\label{eq:anchor}
 R(T,F;T,F).
\end{equation}
We call this additional relation the \emph{anchor relation}. As a Boolean constraint, it is the tautology $T+F=T+F$, so it places no restriction on the anchor values by itself. The purpose of this relation is to control the parity of the number of teeth through its associated ladder. By Lemma~\ref{lem:relation}, that ladder contributes a number of internal teeth whose parity is $T\xor F$.

Each duplicated relation has identical arguments in its two copies, so their tooth counts have the same parity and their combined contribution is even. Consequently, the parity of the total tooth count is determined by the anchor relation. Requiring this total to be odd forces $T\xor F=1$, and therefore $T\ne F$.

All relations are preserved by complementing every logical bit: $a+b=c+d$ implies $(1-a)+(1-b)=(1-c)+(1-d)$. Their parities are preserved as well. Once $T\ne F$ has been established, we may therefore normalize to $T=1,F=0$ by taking the complementary handle if necessary.

\subsubsection{Padding relations}

In this subsection, we introduce the machinery that will be needed to prove Lemma~\ref{lem:small-cuts} in Section~\ref{subsec:correctness}. This lemma will be essential to demonstrate that every tooth of every violated comb in the graph constructed by the reduction must be contained in a ladder.

The logical variables appearing in the relations introduced so far comprise two representatives $X,\bar X$ for each of the $v$ source variables, three auxiliary variables $u_j,v_j,w_j$ for each of the $m$ clauses and the two anchors $T,F$. The total number of variables occurring in the relations is therefore
\begin{equation}\label{eq:q}
 q=2v+3m+2.
\end{equation}
Listing them in any cyclic order $z_1,\ldots,z_q$, we then add the $q$ tautological relations
\begin{equation}\label{eq:padding}
 R(z_i,z_{i+1};z_i,z_{i+1})\quad(1\leq i\leq q),
 \qquad z_{q+1}=z_1.
\end{equation}
We call these relations the \textit{padding relations} and the associated ladders the \textit{padding ladders}. Each padding relation states $z_i+z_{i+1}=z_i+z_{i+1}$, so it holds for every Boolean assignment and does not change satisfiability. By Lemma~\ref{lem:relation}, its associated ladder contributes tooth parity $z_i\xor z_{i+1}$. Since $z_{q+1}=z_1$, each variable occurs twice in the XOR of these parities and therefore cancels:
\begin{equation}\label{eq:cycle-parity}
 \bigxor_{i=1}^{q}(z_i\xor z_{i+1})
 =\bigxor_{i=1}^{q}(z_i\xor z_i)=0.
\end{equation}
Here we use associativity and commutativity of XOR, together with $z\xor z=0$. Thus the padding ladders collectively contribute an even number of teeth and preserve the total tooth parity. Each variable appears in two adjacent padding relations and appears twice in each, adding four occurrences of this variable.

All relations needed for the reduction have now been introduced. Let's now count how many such relations there are. Eq.~\eqref{eq:complement} introduces one complement relation per source variable. Duplicating these gives $2v$ occurrences. Eq.~\eqref{eq:clause-relations} introduces two relations per clause. Duplicating both gives $4m$ occurrences. We then add the single anchor relation of Eq.~\eqref{eq:anchor} and the $q$ padding relations from Eq.~\eqref{eq:padding}. Thus, counting each copy separately, the total number of relation occurrences is
\begin{equation}\label{eq:M}
 M=2v+4m+1+q=4v+7m+3,
\end{equation}
Each relation occurrence will be represented by its own ladder with separate ladders for duplicated relations. We will use the following properties of the resulting relation system:
\begin{enumerate}[label=(\roman*)]
 \item It admits a satisfying assignment with odd total relation parity if and only if the 3-SAT formula $\Phi$ is satisfiable. Indeed, duplicated relations and padding relations contribute even parity, leaving total parity $T\xor F$. Odd parity therefore forces $T\ne F$, allowing normalization to $T=1,F=0$ by complementation of the handle if needed. Under this normalization, Eq.~\eqref{eq:complement} enforces $\bar X=1-X$. Tautological relations from Eq.~\eqref{eq:anchor} and Eq.~\eqref{eq:padding} are always satisfied and Lemma~\ref{lem:clause} guarantees satisfiability of $\Phi$ when both relations in Eq.~\eqref{eq:clause-relations} are satisfied. Conversely, if $\Phi$ is satisfied, we can set $T=1$ and apply the assignments from Eq.~\eqref{eq:clause-witness} to satisfy the relations of Eq.~\eqref{eq:clause-relations}. Eq.~\eqref{eq:complement} now enforces $F=1-T$, giving odd parity, and tautological relations are always satisfied.
 \item Every logical variable $z$ has an even number $N_z\geq6$ of occurrences in the relations. The duplicated non-padding relations give each variable a positive even count, and the padding adds four.
 \item No relation contains the same logical variable more than twice. The only repetitions caused by repeated literals occur in the first clause relation, which has two literal positions. The anchor and padding relations also use each of their two distinct arguments twice.
 \item The bipartite incidence graph between the logical variables and the non-padding relations is connected as illustrated in Fig.~\ref{fig:nonpadding-incidence}. This is the bipartite graph with one vertex for each logical variable, one vertex for each non-padding relation occurrence and an edge between a variable vertex and a relation vertex when the corresponding variable appears in the corresponding relation. Complement relations connect every literal representative to $T,F$. The second relation of each clause connects $u_j,w_j$ to $T$, and the first connects $v_j$ to $u_j$ and the literal representatives.
\end{enumerate}

\begin{figure}[t]
\centering
\begin{tikzpicture}[scale=1.05,
 vertex/.style={circle,draw,fill=white,inner sep=2.4pt},
 relation/.style={rectangle,draw,fill=white,inner sep=2.4pt},
 edge/.style={line width=.6pt},
 group/.style={font=\small\itshape}]

 \node[anchor=east,font=\small] at (-.30,9.50)
  {Logical variables};
 \node[anchor=west,font=\small] at (5.35,9.50)
  {Non-padding relations};

 \node[group,anchor=east] at (-.30,8.70)
  {Anchors};
 \node[group,anchor=east] at (-.30,6.75)
  {Literal representatives};
 \node[group,anchor=east] at (-.30,2.35)
  {Auxiliary variables};

 \coordinate (T) at (0,8.05);
 \coordinate (F) at (0,7.45);

 \coordinate (X)    at (0,6.05);
 \coordinate (Xbar) at (0,5.45);
 \coordinate (Y)    at (0,4.85);
 \coordinate (Ybar) at (0,4.25);
 \coordinate (Z)    at (0,3.65);
 \coordinate (Zbar) at (0,3.05);

 \coordinate (u) at (0,1.70);
 \coordinate (v) at (0,1.10);
 \coordinate (w) at (0,.50);

 \node[group,anchor=west] at (5.35,8.70)
  {Anchor relation};
 \node[group,anchor=west] at (5.35,6.75)
  {Complement relations};
 \node[group,anchor=west] at (5.35,2.35)
  {Clause relations};

 \coordinate (anchor) at (5.20,7.75);

 \coordinate (compX) at (5.20,5.75);
 \coordinate (compY) at (5.20,4.55);
 \coordinate (compZ) at (5.20,3.35);

 \coordinate (clauseone) at (5.20,1.55);
 \coordinate (clausetwo) at (5.20,.65);

 \foreach \z in {X,Xbar,T,F}
  \draw[edge] (\z)--(compX);
 \foreach \z in {Y,Ybar,T,F}
  \draw[edge] (\z)--(compY);
 \foreach \z in {Z,Zbar,T,F}
  \draw[edge] (\z)--(compZ);

 \foreach \z in {X,Y,u,v}
  \draw[edge] (\z)--(clauseone);
 \foreach \z in {u,Zbar,T,w}
  \draw[edge] (\z)--(clausetwo);

 \foreach \z in {T,F}
  \draw[edge] (\z)--(anchor);

 \node[vertex,label=left:$X$] at (X) {};
 \node[vertex,label=left:$\overline X$] at (Xbar) {};
 \node[vertex,label=left:$Y$] at (Y) {};
 \node[vertex,label=left:$\overline Y$] at (Ybar) {};
 \node[vertex,label=left:$Z$] at (Z) {};
 \node[vertex,label=left:$\overline Z$] at (Zbar) {};

 \node[vertex,label=left:$u_1$] at (u) {};
 \node[vertex,label=left:$v_1$] at (v) {};
 \node[vertex,label=left:$w_1$] at (w) {};

 \node[vertex,label=left:$T$] at (T) {};
 \node[vertex,label=left:$F$] at (F) {};

 \node[relation,label=right:{$R(X,\overline X;T,F)$}]
  at (compX) {};
 \node[relation,label=right:{$R(Y,\overline Y;T,F)$}]
  at (compY) {};
 \node[relation,label=right:{$R(Z,\overline Z;T,F)$}]
  at (compZ) {};

 \node[relation,label=right:{$R(X,Y;u_1,v_1)$}]
  at (clauseone) {};
 \node[relation,label=right:{$R(u_1,\overline Z;T,w_1)$}]
  at (clausetwo) {};

 \node[relation,label=right:{$R(T,F;T,F)$}]
  at (anchor) {};

\end{tikzpicture}
\caption{
The bipartite incidence graph between the logical variables and the non-padding relations for the single-clause formula $\Phi=(x\lor y\lor\neg z)$, whose literal representatives are $X,Y,\overline Z$. We therefore use $(a,b,c)=(X,Y,\overline Z)$ for the clause relations $R(a,b;u_1,v_1)$ and $R(u_1,c;T,w_1)$. Logical variables on the left are grouped into anchors, literal representatives and auxiliary variables. Relations on the right are grouped into the anchor relation, complement relations and clause relations. Only one copy of each duplicated relation is shown. An edge indicates that a variable occurs in a relation. Every vertex has a path to $T$, so the graph is connected.}
\label{fig:nonpadding-incidence}
\end{figure}
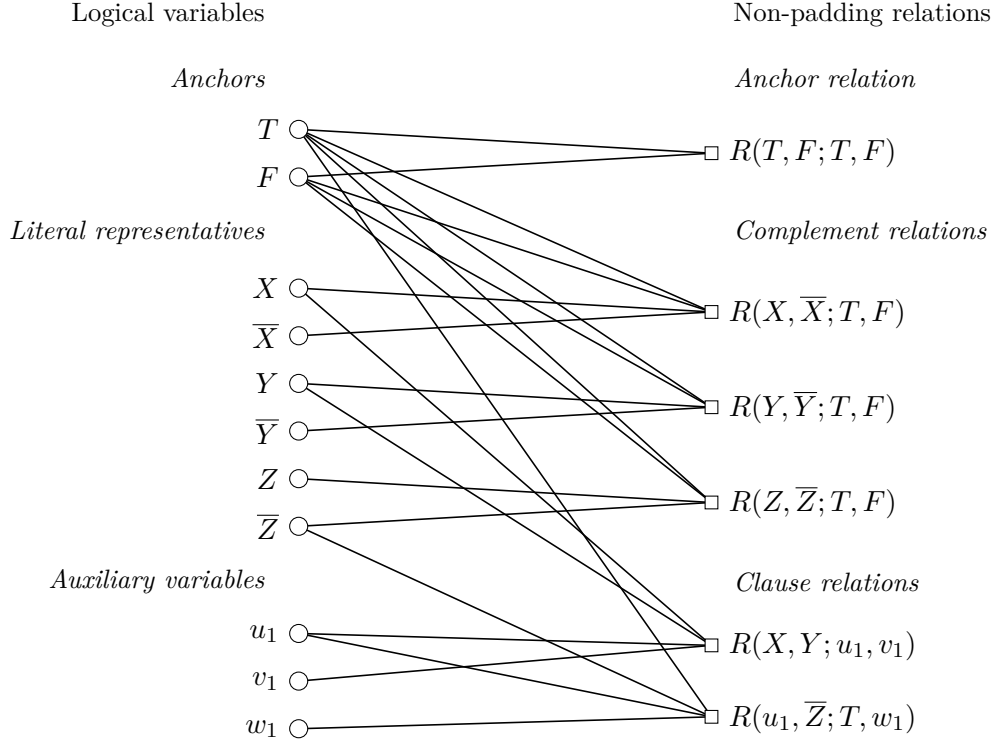

\subsubsection{Encoding variables with prisms}

Before describing the variable gadget, we introduce the graph notation used to define it. The symbol $\square$ denotes the Cartesian product of graphs~\cite{Hammack2011}. For graphs $G_1=(V_1,E_1)$ and $G_2=(V_2,E_2)$, the product $G_1\square G_2$ has vertex set $V_1\times V_2$. In this context, the graphs $G_1$ and $G_2$ are called the factors. Two vertices $(u,a)$ and $(v,b)$ of $G_1\square G_2$ are adjacent precisely when either $u=v$ and $\{a,b\}\in E_2$ or $a=b$ and $\{u,v\}\in E_1$.

Let $C_k$ denote a cycle on $k\geq3$ vertices and let $K_2$ denote the graph consisting of two vertices joined by one edge. The prism graph $C_k\square K_2$ consists of two copies of $C_k$, with an additional edge joining each vertex in one copy to its corresponding vertex in the other. For example, $C_3\square K_2$ is the triangular prism graph, as illustrated in Fig.~\ref{fig:cartesian-product-C3-K2}, and $C_4\square K_2$ is the graph of a cube. In general, a prism has $2k$ vertices and $3k$ edges. Each vertex has two neighbors along its cycle and one in the other copy, so it is a cubic graph, meaning that every vertex has degree three.

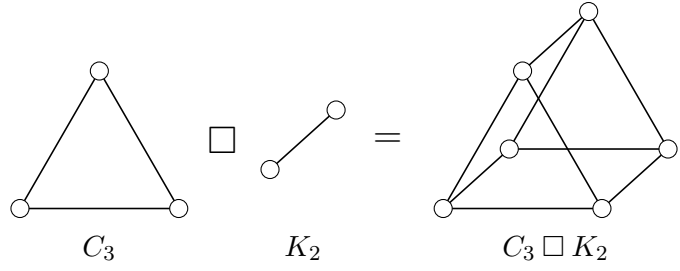
\begin{figure}[t]
\centering
\begin{tikzpicture}[scale=1.05,
 vertex/.style={circle,draw,fill=white,inner sep=2.4pt},
 edge/.style={line width=.6pt}]

 \def\edgelen{2}
 \def\bridgelen{1.12}
 \pgfmathsetmacro{\triheight}{sqrt(3)*\edgelen/2}
 \pgfmathsetmacro{\dx}{\bridgelen*cos(42)}
 \pgfmathsetmacro{\dy}{\bridgelen*sin(42)}
 \pgfmathsetmacro{\kx}{\edgelen+1.15}
 \pgfmathsetmacro{\px}{\kx+\dx+1.35}

 \coordinate (c1) at (0,0);
 \coordinate (c2) at (\edgelen,0);
 \coordinate (c3) at ({\edgelen/2},\triheight);
 \draw[edge] (c1)--(c2)--(c3)--cycle;
 \foreach \i in {1,2,3}
  \node[vertex] at (c\i) {};
 \node at ({\edgelen/2},-.5) {$C_3$};

 \node[font=\Large]
  at ({\edgelen+.55},{\triheight/2}) {$\square$};

 \coordinate (k1) at (\kx,{(\triheight-\dy)/2});
 \coordinate (k2) at ({\kx+\dx},{(\triheight+\dy)/2});
 \draw[edge] (k1)--(k2);
 \foreach \j in {1,2}
  \node[vertex] at (k\j) {};
 \node at ({\kx+\dx/2},-.5) {$K_2$};

 \node[font=\Large]
  at ({\kx+\dx+.65},{\triheight/2}) {$=$};

 \begin{scope}[shift={(\px,0)}]
  \foreach \j/\xx/\yy in {1/0/0,2/\dx/\dy}{
   \coordinate (p1\j) at (\xx,\yy);
   \coordinate (p2\j) at ({\xx+\edgelen},\yy);
   \coordinate (p3\j)
    at ({\xx+\edgelen/2},{\yy+\triheight});
  }

  \foreach \i in {1,2,3}
   \draw[edge] (p\i1)--(p\i2);

  \foreach \j in {1,2}
   \draw[edge] (p1\j)--(p2\j)--(p3\j)--cycle;

  \foreach \i in {1,2,3}
   \foreach \j in {1,2}
    \node[vertex] at (p\i\j) {};

  \node at ({(\edgelen+\dx)/2},-.5)
   {$C_3\mathbin{\square}K_2$};
 \end{scope}

\end{tikzpicture}
\caption{Illustration of the Cartesian product $C_3\mathbin{\square}K_2$ yielding a triangular prism graph.}
\label{fig:cartesian-product-C3-K2}
\end{figure}

For each logical variable $z$, we now take the prism $Q_z=C_{N_z/2}\square K_2$, where $N_z$ is the number of occurrences of the variable in all relations. Setting $k=N_z/2$, label the vertices of $Q_z$ with $(i,j)$, where $i\in\{0,\ldots,k-1\}$ and $j\in\{0,1\}$. The index $i$ specifies a position on a cycle, while $j$ selects one of the two copies. Its edges join $(i,j)$ to $(i+1\bmod k,j)$ and join $(i,0)$ to $(i,1)$. Since $N_z$ is even and at least six, this is a cubic graph with exactly $N_z$ vertices, one for each occurrence of $z$. Give each prism edge weight one half. In this case, if the prism is part of any violated comb, then all its vertices will carry the same handle color, as proved in Section~\ref{subsec:correctness}.

For each variable $z$, assign each of the $N_z$ occurrences of $z$ in the relations to a distinct vertex of the prism $Q_z$. For each relation occurrence, introduce a six-vertex ladder. Connect the prism vertices assigned to the four arguments of its relation $(a,b,c,d)$ to its ports $(a_1,b_3,a_3,b_1)$, respectively, with half-weight edges. Repeated arguments occupy distinct prism vertices. Give all unspecified edges of the complete graph weight zero. This construction uses only the propositional logic formula and does not require knowing whether it is satisfiable.

Each prism vertex has three internal half-weight edges and one external half-weight edge. Each ladder port has one unit rung, one rail half-weight edge and one external half-weight edge. Each middle ladder vertex has one unit rung and two rail half-weight edges. All weighted degrees are therefore two, and the maximum number of incident edges for a vertex is four.

There are four arguments per relation, hence $4M$ prism vertices and $6M$ ladder vertices. The total number of vertices in the graph constructed from the prism and ladder gadgets is therefore
\begin{equation}\label{eq:counts}
 n=10M=40v+70m+30.
\end{equation}
The prisms contribute $6M$ half-weight edges, the ladders contribute $4M$ rail half-weight edges and $3M$ unit edges and the external connections between the prisms and ladders contribute $4M$ half-weight edges. Thus the graph has $14M$ half-weight edges and $3M$ unit edges. The number of vertices and positive edges is linear in $v+m$. Writing the full $\binom{n}{2}$-coordinate vector takes quadratic space in $n$, and is still polynomial.

\subsection{Correctness of the reduction}\label{subsec:correctness}

\subsubsection{Classification of tight cuts}

The degree equations alone do not guarantee subtour feasibility. We now prove a stronger statement that will show $x^* \in \SEP(n)$ for the constructed graph and also restrict every tight set used by a violated comb.

\begin{definition}
    The \textit{shores} of a cut $\delta(S)$ of a connected graph are the sets $S$ and $\overline{S}:=V \setminus S$ \cite{Frank2011}.
\end{definition}

\begin{lemma}\label{lem:small-cuts}
For every nonempty proper subset $S\subset V$ of the vertex set of the constructed graph, $x^*(\delta(S))\geq 2$. If equality holds, at least one of the shores $S$ and $V\setminus S$ consists of a single prism vertex or is contained in a single ladder.
\end{lemma}
\begin{proof}
Fix a nonempty proper subset $S$ and color its vertices one, with all remaining vertices colored zero. A prism or ladder is \emph{split} if it contains both colors. Partition the support edges into disjoint groups: the interior edges of each prism, and the interior edges together with all four external edges of each ladder. All lower bounds below refer to these disjoint edge groups, so their contributions can be added.

First, a split prism contributes a cut weight of at least $3/2$ internally. To verify this, it is enough to show that deleting at most two edges cannot disconnect a cubic prism. If neither cycle loses more than one edge, both remain connected paths or cycles, and at least one of the $k\geq3$ joining edges remains. If both deleted edges belong to the same cycle, the other cycle and every joining edge remain, so every component of the broken cycle attaches to the intact one. These cases exhaust all deletions of at most two edges.

Second, a split ladder contributes a cut weight of at least one internally. A split rung already contributes a cut weight of one. If no rung is split, the three whole rungs are monochromatic. Since the ladder is split, consecutive rungs of opposite colors exist, and their two rail edges contribute a cut weight of one.

Third, consider a padding relation joining two whole prisms of different colors. If its ladder is split, the preceding bound gives a contribution of at least one. If its ladder is monochromatic, exactly two of its four external edges cross, again contributing one. Finally, any non-padding relation incident to whole prisms of both colors contributes at least $1/2$. Indeed, a monochromatic ladder cuts at least one external edge, while a split ladder contributes at least one internally.

Assume now that $x^*(\delta(S))\leq2$. We show that equality must hold and that one of the two shores has the stated form. Since each split prism contributes at least $3/2$ to the cut, at most one prism is split. Suppose first that exactly one prism is split. No ladder can be split since this would give a cut weight of at least $3/2+1>2$. The padding ladders, together with the prisms, form a cycle. See Fig.~\ref{fig:padding-cut}~(a) for example. If the remaining whole prisms have different colors, then deleting the split prism from this padding cycle would leave a path containing prisms of both colors. Two consecutive prisms on this path must therefore have different colors, and the padding ladder joining them contributes at least one to the cut. Together with the split prism's contribution, this gives $x^*(\delta(S))\geq3/2+1>2$, a contradiction. Thus, if a prism is split, no ladder can be split and all remaining prisms have the same color.

Suppose next that no prism is split and the whole prisms have different colors. As illustrated in Fig.~\ref{fig:padding-cut} (b), traversing the padding cycle encounters at least two color transitions, whose distinct padding ladders contribute together at least two. Connectedness of the non-padding incidence graph also guarantees a non-padding relation incident to both colors. Otherwise every relation would connect only same-colored variables, contradicting connectedness. That additional ladder contributes at least $1/2$. Adding this contribution to the contribution of the padding ladders contradicts the assumption $x^*(\delta(S))\leq2$.

\begin{figure}[t]
\centering
\begin{tikzpicture}[scale=1.05,
 vertex/.style={circle,draw,fill=white,inner sep=2.4pt,
                minimum size=9pt},
 edge/.style={line width=.6pt},
 ladder/.style={rectangle,draw,line width=.5pt,fill=white,
                inner sep=3pt,font=\small,
                minimum width=9mm,minimum height=5.5mm},
 crossing/.style={ladder,fill=black!15},
 split/.style={vertex,path picture={
   \fill[black]
    (path picture bounding box.south)
    rectangle (path picture bounding box.north east);
 }}]

 \foreach \panel/\xx/\heading in
  {1/0/{(a) One split prism},2/7.80/{(b) All prisms whole}}{
  \begin{scope}[shift={(\xx,0)}]
   \node[font=\small] at (0,3.55) {\heading};

   \coordinate (q1) at (0,2.50);
   \coordinate (q2) at (2.165,1.25);
   \coordinate (q3) at (2.165,-1.25);
   \coordinate (q4) at (0,-2.50);
   \coordinate (q5) at (-2.165,-1.25);
   \coordinate (q6) at (-2.165,1.25);

   \foreach \i/\j in {1/2,2/3,3/4,5/6,6/1}
    \draw[edge] (q\i)--(q\j);

   \ifnum\panel=2
    \node[crossing] (r) at (0,0) {$R$};
    \draw[edge,dashed] (r.north east)--(q2);
    \draw[edge,dashed] (r.south west)--(q5);
    \draw[edge,dashed] (r.north west)--++(-.40,.30);
    \draw[edge,dashed] (r.south east)--++(.40,-.30);
   \fi

   \node[ladder] at (1.0825,1.875) {$P_1$};
   \node[ladder] (p4) at (-1.0825,-1.875) {$P_4$};
   \node[ladder] at (-1.0825,1.875) {$P_q$};
   \ifnum\panel=1
    \node[ladder] at (2.165,0) {$P_2$};
    \node[crossing] at (1.0825,-1.875) {$P_3$};
    \node[ladder] at (-2.165,0) {$P_{q-1}$};
   \else
    \node[crossing] at (2.165,0) {$P_2$};
    \node[ladder] at (1.0825,-1.875) {$P_3$};
    \node[crossing] at (-2.165,0) {$P_{q-1}$};
   \fi

   \ifnum\panel=1
    \node[split] (v1) at (q1) {};
    \foreach \i in {2,3}
     \node[vertex,fill=black] (v\i) at (q\i) {};
    \foreach \i in {4,5,6}
     \node[vertex] (v\i) at (q\i) {};
   \else
    \foreach \i in {1,2,6}
     \node[vertex,fill=black] (v\i) at (q\i) {};
    \foreach \i in {3,4,5}
     \node[vertex] (v\i) at (q\i) {};
   \fi

   \draw[edge] (v4)--(p4);
   \draw[edge] (p4)--
    node[midway,sloped,fill=white,inner sep=.2pt]
     {$\mathinner{\cdotp\mkern-2mu\cdotp\mkern-2mu\cdotp}$}
    (v5);

   \node[anchor=south] at (0,2.73) {$Q_{z_1}$};
   \node[anchor=west] at (2.40,1.25) {$Q_{z_2}$};
   \node[anchor=west] at (2.40,-1.25) {$Q_{z_3}$};
   \node[anchor=north] at (0,-2.73) {$Q_{z_4}$};
   \node[anchor=east] at (-2.40,-1.25) {$Q_{z_{q-1}}$};
   \node[anchor=east] at (-2.40,1.25) {$Q_{z_q}$};
  \end{scope}
 }

\end{tikzpicture}
\caption{Two configurations excluded by a cut of weight at most two. Circles represent variable prisms, boxes represent ladders, and black and white indicate membership in $S$ and its complement. Solid lines describe the padding cycle with the $P_i$ ladders being padding ladders. Dashed lines in (b) show two variable neighbors of the non-padding ladder $R$ (the two dangling dashed lines are attached to prisms that are not displayed in the figure). Grey boxes mark the ladders contributing to the cut lower bound. In (a), the split prism $Q_{z_1}$ contributes at least $3/2$ to the cut, and the padding ladder $P_3$ between oppositely colored whole prisms contributes at least one to the cut. In (b), the two padding ladders $P_2,P_{q-1}$ contribute at least one each, and $R$ contributes at least $1/2$ to the cut. A ladder's contribution includes its internal and external edges. In each panel, the contributions come from disjoint edge groups and give $x^*(\delta(S))\geq5/2>2$.}
\label{fig:padding-cut}
\end{figure}
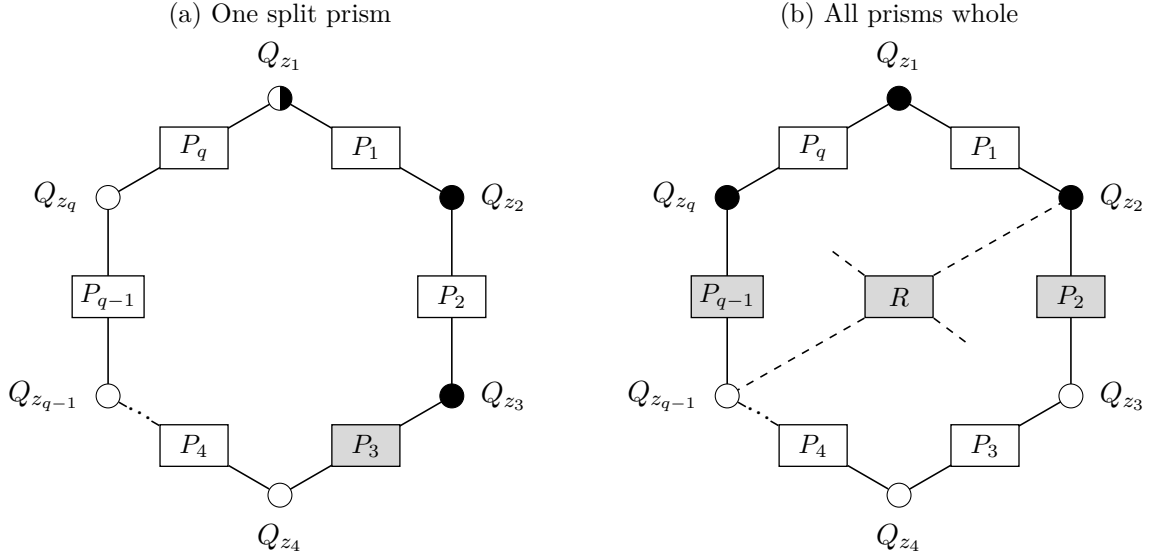

In either case, all whole prisms have a common color. Complement $S$ if necessary so that they all lie outside $S$. This leaves $\delta(S)$ unchanged. There are now two cases.

If no prism is split, $S$ is a union of portions of ladders. For each nonempty portion, all external endpoints of the ladder are outside $S$. Therefore, the cut contribution of a portion is exactly its cut weight in the rooted ladder, which is at least two by Lemma~\ref{lem:ladder-tight}. To satisfy $x^*(\delta(S))\leq2$, only one ladder can meet $S$, and its contribution must equal two.

If a prism is split, let $P$ be its nonempty intersection with $S$. No ladder can be split, so every ladder is monochromatic. Since a relation contains a logical variable at most twice, a ladder $U$ has $k_U\leq2$ external vertex neighbors in $P$ and $4-k_U$ external vertex neighbors outside of $S$. Its external contribution is either $k_U/2$ when outside $S$ or $(4-k_U)/2$ when inside. Both are at least $k_U/2$. Each vertex of $P$ has exactly one external edge, so $\sum_U k_U=|P|$, giving
\begin{equation}\label{eq:singleton-bound}
 x^*(\delta(S))\geq\frac32+\frac{|P|}{2}.
\end{equation}
Requiring $x^*(\delta(S)) \leq 2$ forces $|P|=1$. Its prism contribution is exactly $3/2$, and the external lower bound is $1/2$. Now every $k_U\leq1$. Placing any ladder inside $S$ increases its external cut contribution from $k_U/2$ to $(4-k_U)/2\geq 3/2$, exceeding the available cut budget. All ladders therefore lie outside $S$, so $S$ is precisely the single prism vertex and its cut has weight two.
\end{proof}

Every vertex has weighted degree two and every edge weight is nonnegative. Lemma~\ref{lem:small-cuts} excludes a cut from a nonempty proper subset to have a weight below two, proving that the fractional point $x^*$ represented by the constructed graph is in $\SEP(n)$. This conclusion follows from the constructed graph; subtour feasibility was not assumed in the proof.

\subsubsection{Excluding teeth that span several gadgets}

We now address the remaining obstacle to show that violated combs in the constructed graph yield a satisfied propositional logic formula. Call each six-vertex ladder a block, and also call each individual prism vertex a block. These blocks partition $V$, the set of vertices of the constructed graph. By Lemma~\ref{lem:small-cuts}, every tight set is contained in one block or has its complement contained in one block.

\begin{lemma}\label{lem:local-teeth}
Each tooth of every violated comb in the constructed graph is contained in a single ladder.
\end{lemma}
\begin{proof}
The constructed fractional point $x^*$ is half-integral and, as established above using Lemma~\ref{lem:small-cuts}, belongs to $\SEP(n)$. Hence Lemma~\ref{lem:equality} applies: each tooth of a violated comb and both its halves are tight. By Lemma~\ref{lem:small-cuts} again, each of these sets is contained in a block or has its complement contained in a block. We show that the second case cannot occur for a tooth.

The following is illustrated in Fig.~\ref{fig:tooth-contraction} (a). Suppose some tooth $T_i$ is not contained in a block. Then $T_i = V \setminus S$ with $\varnothing \neq S \subseteq U$ for some block $U$. Two consequences follow. First, every other tooth is disjoint from $T_i$ and therefore lies in $S \subseteq U$. Second, $V \setminus U \subseteq T_i$, so $T_i$ contains every block other than $U$, which is at least three blocks.

Since two halves of a tooth each contained in a single block cover at most two blocks, at least one half of $T_i$ is not contained in a block. Replacing $H$ by $V\setminus H$ preserves the comb conditions and the left-hand side of the comb inequality and has as only effect to exchange the two halves of every tooth. We may therefore assume this half to be $A_i = T_i \cap H$ and define as usual $B_i = T_i \setminus H$. The following argument applies even if $A_i$ and $B_i$ are interchanged. By Lemma~\ref{lem:equality}, $A_i$ must be tight, so by Lemma~\ref{lem:small-cuts} we must have $V \setminus A_i \subseteq W$ for some block $W$. The inclusion $A_i\subseteq T_i$ then yields $S=V\setminus T_i\subseteq V\setminus A_i \subseteq W$. From $S \subseteq W$ and $S \subseteq U$, the nonempty set $S$ lies in $U \cap W$. As distinct blocks are disjoint, we must have $W = U$. Consequently, $V\setminus A_i\subseteq U$, or equivalently, $V\setminus U\subseteq A_i$. Hence
\begin{equation}
  V \setminus U \subseteq A_i \subseteq H
  \quad\text{and}\quad
  B_i \subseteq U \setminus S .
\end{equation}

\begin{figure}[t]
\centering
\begin{tikzpicture}[scale=1.05,
 vertex/.style={circle,draw,fill=white,inner sep=2.4pt},
 half/.style={line width=.6pt},
 unit/.style={line width=1.7pt},
 boundary/.style={line width=.5pt,rounded corners=5pt},
 tooth/.style={line width=.5pt,rounded corners=6pt}]

 \node[font=\small] at (2.70,5.02)
  {(a) Hypothetical tooth spanning blocks};

 \filldraw[boundary,fill=black!10]
  (0,0) rectangle (5.40,4.35);
 \node at (5.00,4.08) {$V$};

 \draw[boundary] (.28,.16) rectangle (5.12,2.82);
 \node at (4.84,2.57) {$U$};

 \node at (2.45,3.86) {$V\setminus U\subseteq A_i$};
 \foreach \xx in {.78,1.70,3.42,4.34}
  \node[vertex,fill=black] at (\xx,3.26) {};
 \node at (2.57,3.26) {$\cdots$};

 \filldraw[boundary,fill=white]
  (.65,.42) rectangle (2.05,2.25);
 \node at (1.35,1.95) {$B_i$};
 \node[vertex] at (1.35,1.43) {};
 \node[vertex] at (1.35,1.01) {};

 \filldraw[boundary,fill=white]
  (2.70,.42) rectangle (4.57,2.25);
 \node[font=\small] at (3.635,2.00)
  {$S=V\setminus T_i$};

 \foreach \xx/\name in {3.16/j,4.02/k}{
  \draw[tooth]
   ({\xx-.23},.99) rectangle ({\xx+.23},1.73);
  \node[vertex,fill=black] at (\xx,1.49) {};
  \node[vertex] at (\xx,1.22) {};
  \node[font=\small] at (\xx,.77) {$T_\name$};
 }

 \node at (2.70,-.46)
  {$T_i=A_i\cup B_i=V\setminus S$};

 \draw[->,line width=.7pt]
  (5.76,2.30)--(7.23,2.30);
 \node[anchor=south,align=center,font=\small]
  at (6.50,2.47)
  {contract\\[2pt]$V\setminus U$\\[2pt]to $r$};

 \begin{scope}[shift={(8.45,.45)}]
  \node[font=\small] at (2,4.57)
   {(b) The rooted ladder $L^\circ$};

  \foreach \i/\xx in {1/0,2/2,3/4}{
   \coordinate (a\i) at (\xx,1.60);
   \coordinate (b\i) at (\xx,0);
  }
  \coordinate (r) at (2,3.55);

  \draw[line width=.5pt]
   (2.00,4.08)
   .. controls (-1.38,4.08) and (-1.72,-.62)
      .. (-.26,-.77)
   .. controls (.16,-.813) and (.64,-.65)
      .. (.51,-.25)
   .. controls (.10,1.00) and (.88,2.50)
      .. (2.23,3.30)
   .. controls (2.70,3.58) and (2.66,4.08)
      .. (2.00,4.08)
   -- cycle;

  \draw[tooth]
   (1.68,-.70) rectangle (2.32,2.15);
  \draw[tooth]
   (3.68,-.70) rectangle (4.32,2.15);

  \draw[half] (r)--(a1);
  \draw[half] (r)--(a3);

  \draw[half] (r)
   .. controls (-1.10,3.30) and (-1.10,-.65)
      .. (b1);
  \draw[half] (r)
   .. controls (5.10,3.30) and (5.10,-.65)
      .. (b3);

  \foreach \i in {1,2,3}
   \draw[unit] (a\i)--(b\i);

  \foreach \i/\j in {1/2,2/3}{
   \draw[half] (a\i)--(a\j);
   \draw[half] (b\i)--(b\j);
  }

  \node[vertex,label=above:$a_1$] at (a1) {};

  \foreach \i in {2,3}
   \node[vertex,fill=black,label=above:$a_\i$]
    at (a\i) {};

  \foreach \i in {1,2,3}
   \node[vertex,label=below:$b_\i$]
    at (b\i) {};

  \node[vertex,fill=black,label=above:$r$]
   at (r) {};

  \node[anchor=west,inner sep=0pt]
   at (2.61,3.96) {$T'_i$};
  \node[anchor=east,inner sep=0pt]
   at (1.59,.80) {$T_j$};
  \node[anchor=east,inner sep=0pt]
   at (3.59,.80) {$T_k$};

  \node[font=\small] at (2.38,-1.03)
   {$T'_i=(U\setminus S)\cup\{r\}$};
 \end{scope}

\end{tikzpicture}
\caption{Contraction of a hypothetical tooth spanning several blocks. Panel (a) shows set inclusions: $T_i=V\setminus S$, all vertices outside the block $U$ belong to $A_i$, and the other teeth lie in $S\subseteq U$. $A_i$ is represented by the grey area. Black vertices belong to the handle and white vertices do not. Contracting $V\setminus U$ to $r$ produces the rooted ladder in (b). The illustrated case has $T'_i=\{r,a_1,b_1\}$, $T_j=\{a_2,b_2\}$ and $T_k=\{a_3,b_3\}$.}
\label{fig:tooth-contraction}
\end{figure}
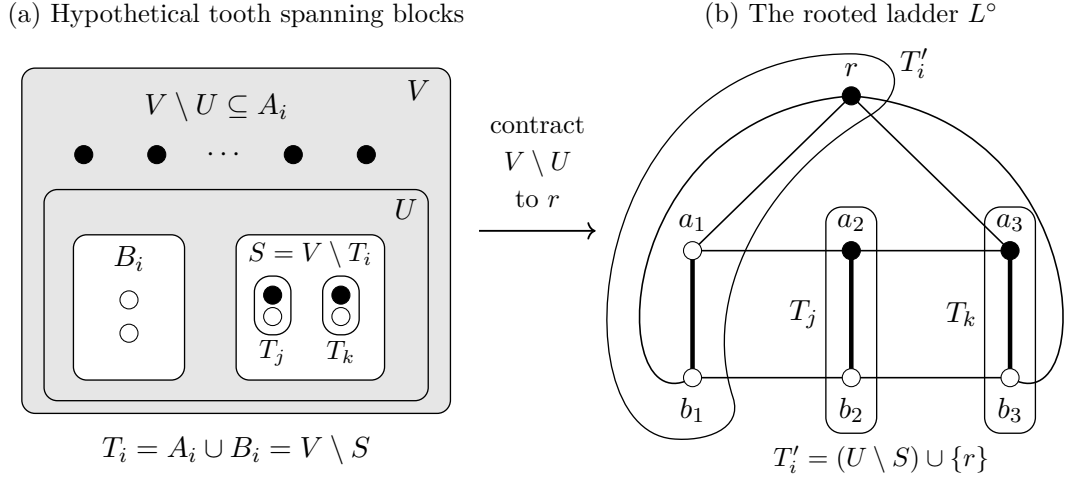

The block $U$ contains the nonempty set $B_i$ and the $s - 1 \geq 2$ other teeth, each with at least two vertices, so $|U| \geq 5$. Since every block is either a single prism vertex or a six-vertex ladder, $U$ must be a ladder.

Given that $U$ is a ladder, contract all vertices outside $U$ to a single root vertex. See Fig.~\ref{fig:tooth-contraction} (b). This simplification can be made because every set appearing in the comb either contains all outside vertices or none of them. This holds for the handle, the large tooth, and every remaining tooth. Hence all their cut weights are preserved by this contraction. The large tooth keeps two nonempty halves because its half containing the outside vertices keeps the root, while the other half lies in $U$. The other teeth are unchanged. Tooth disjointness and their odd number are preserved as well. The contracted weighted graph is exactly $L^\circ$ and we have obtained a violated comb in $L^\circ$, contrary to Lemma~\ref{lem:rooted-free}.

Hence every tooth lies in a single block. A prism-vertex block cannot hold a tooth with two nonempty halves, so every tooth lies in a ladder.
\end{proof}

\subsubsection{Nonplanar support}

\begin{lemma}
\label{lem:nonplanar}
Replace the prism associated with the auxiliary variable $v_1$ by the bipartite graph $K_{3,3}$ on the same six vertices, assigning weight one half to every internal edge and retaining all external connections. The resulting point is half-integral and belongs to $\SEP(n)$. Its support is nonplanar and has maximum degree four. The conclusions of Lemmas~\ref{lem:small-cuts} and~\ref{lem:local-teeth} remain valid, with ``prism vertex'' interpreted as ``variable-gadget vertex''. The numbers of vertices and positive edges are unchanged.
\end{lemma}

\begin{proof}
The auxiliary variable $v_1$ occurs once in the first relation encoding the first clause of the 3-SAT formula. Duplication gives two occurrences of this variable, and the padding relations add four. Hence $N_{v_1}=6$, so its original gadget is a triangular prism.

Partition these six vertices into two sets of size three and replace the internal prism edges by all nine edges between the two sets. This gives the complete bipartite graph $K_{3,3}$ which is known to be nonplanar~\cite{Kuratowski1930}. Retain the variable occurrence assigned to each vertex and its external half-weight edge. Both the triangular prism and $K_{3,3}$ are cubic graphs with six vertices and nine edges. Consequently, the replacement preserves half-integrality, all weighted degrees, the maximum support degree, and the
graph-size counts.

We verify the internal cut bound needed in Lemma~\ref{lem:small-cuts}. Consider a nonempty proper subset $S$ in the replacement gadget. Complementing it within the gadget if necessary, assume that it contains at most three vertices. If it contains $a$ vertices from one partite set and $b$ from the other, then $1\leq a+b\leq3$. Each of the $a$ vertices in one partite set has $3-b$ neighbors outside $S$ and each of the $b$ vertices in the other has $3-a$. Thus, the internal cut of $S$ contains $a(3-b)+b(3-a)=3(a+b)-2ab \geq 3$ edges and therefore has weight at least $3/2$.

Let $\widehat{x}^*$ denote the modified fractional point, and consider a nonempty proper subset $S$ with $\widehat{x}^*(\delta(S))\leq2$. Every split variable gadget contributes at least $3/2$ internally. All external edges remain untouched so all arguments from Lemma~\ref{lem:small-cuts} still apply with ``prism'' interpreted as ``variable gadget'' and ``prism vertex'' interpreted as ``variable-gadget vertex''. Consequently, the cut classification of this lemma remains valid for the modified construction. Together with nonnegativity and the degree equations, this proves $\widehat{x}^*\in\SEP(n)$.

Lemma~\ref{lem:local-teeth} then applies using this cut classification, with each vertex of the replacement gadget treated as a singleton block. The rooted-ladder argument is unchanged. Finally, the support contains $K_{3,3}$ as a subgraph and is therefore nonplanar.
\end{proof}

\subsubsection{Both directions and NP membership}

\begin{proof}[Proof of Theorem~\ref{thm:main}]
Construct the fractional point $x^*$ from the 3-SAT formula $\Phi$ as in Section~\ref{subsec:construction}. We have proved that it is half-integral and belongs to $\SEP(n)$, with maximum support degree four and polynomial encoding size. The support graph can be explicitly made nonplanar by Lemma~\ref{lem:nonplanar} with Lemmas~\ref{lem:small-cuts} and~\ref{lem:local-teeth} remaining applicable.

Suppose first that a comb is violated. By Lemma~\ref{lem:local-teeth}, all selected semicuts lie inside ladders. Zero residual imposed by Lemma~\ref{lem:equality} then forces every external edge and every internal edge of a variable gadget to be uncut. Every variable gadget is connected, so all its vertices have a common handle color. This color defines the Boolean value of the corresponding logical variable, shared by all its occurrences. The external edges transfer these colors to the
attached ports. Within each ladder, Lemma~\ref{lem:relation} yields its satisfied relation and its tooth parity.

The total tooth count is odd. Duplicate relation parities cancel in pairs and padding parities cancel by Eq.~\eqref{eq:cycle-parity}. The remaining anchor contribution is $T\xor F$, so $T\ne F$. Complementing the handle if necessary makes $T=1,F=0$. Eq.~\eqref{eq:complement} now enforces every literal negation, and Lemma~\ref{lem:clause} guarantees that every clause of $\Phi$ is satisfied.

Conversely, suppose $\Phi$ is satisfiable. Set $T=1,F=0$, use corresponding values for $X$ and $\bar X$ from the 3-SAT assignment, and choose the auxiliary values from Eq.~\eqref{eq:clause-witness}. All relations hold. Introduce and connect the variable gadgets and ladders related to the logical variables and relations. The support graph has maximum degree four. Color each variable gadget and its attached ports uniformly with the logical bit associated to their corresponding variable to have external and variable gadget edges uncut. In each ladder, choose the local realization supplied by Lemma~\ref{lem:relation} and Table~\ref{tab:states}. The selected teeth are pairwise disjoint and each lies in a single ladder. By Lemma~\ref{lem:relation} the two halves of every tooth form a tight domino and their semicuts cover exactly the handle cut, giving zero residual.

The total tooth count is odd by the duplication, padding, and anchor construction. Moreover, $R(T,F;T,F)=R(1,0;1,0)$ uses the opposite-rail state and contributes three rung teeth by itself. There are therefore at least three teeth. All comb conditions hold, and Eq.~\eqref{eq:slack-identity} gives slack $-1$.

We have proved the required equivalence
\begin{equation}\label{eq:reduction}
 \Phi\text{ is satisfiable}
 \quad\Longleftrightarrow\quad
 x^*(\Phi)\text{ violates a general comb inequality}.
\end{equation}
Since 3-SAT is NP-complete, this proves NP-hardness. Finally, a comb can be used as certificate and has at most $\lfloor n/2\rfloor$ teeth, so its handle and teeth have polynomial-size descriptions. Checking the comb conditions and computing its slack can therefore be done in polynomial time. If subtour feasibility and half-integral points are treated as an input promise, these checks suffice. Alternatively, we may define the decision problem to reject inputs that are outside $\SEP(n)$ or are not half-integral. In this formulation, the verifier checks half-integrality in polynomial time and also checks $x^*\in\SEP(n)$, which is possible in polynomial time as described in Section~\ref{subsec:SEP}. The decision problem is therefore in $\NP$. Nonplanarity and the support-degree bound are also checkable in polynomial time~\cite{Hopcroft1974}, so the same certificate argument establishes membership in NP when these restrictions are included in the definition of the decision language.
\end{proof}

\subsubsection{Proof of the remaining corollary}\label{subsubsec:cor_proofs}

\begin{proof}[Proof of Corollary~\ref{cor:small-teeth}]
For a satisfiable formula, Table~\ref{tab:states} constructs the violated comb using only rung teeth and teeth $P_1\mid P_2$, whose two halves each have size one in the first case and size two in the second case. For an unsatisfiable formula, Eq.~\eqref{eq:reduction} shows that $x^*(\Phi)$ violates no general comb inequality, and hence no inequality in this restricted family.
\end{proof}

\section{Conclusion}\label{sec:conclusion}

This work resolves a longstanding open question in polyhedral combinatorics: the computational complexity of comb separation. Although comb inequalities have played an important role in exact TSP methods since their introduction in their general form in 1979 by Gr\"otschel and Padberg~\cite{Grotschel1979a,Grotschel1979b}, neither a polynomial-time exact separation algorithm nor an NP-hardness proof was known~\cite{applegate2007, Cook2019}. Our reduction from 3-SAT establishes that deciding whether a violated comb exists is NP-complete and that exact comb separation is NP-hard. Consequently, unless $\PP=\NP$, the polynomial-time separation algorithm conjectured by Padberg and Gr\"otschel in 1985~\cite{padberg1985} does not exist. The hardness persists under strong restrictions. The constructed vectors belong to the subtour elimination polytope, have coordinates in $\{0,\tfrac12,1\}$, and have support graphs of maximum degree four that can explicitly be made nonplanar. Moreover, the decision problem remains NP-complete when admissible teeth contain only two or four vertices, divided equally between the handle and its complement.

The result also has consequences beyond exact separation. Every violated comb at a half-integral point of $\SEP(n)$ is maximally violated. Unless $\PP=\NP$, no polynomial-time algorithm can therefore always estimate the maximum violation with absolute error strictly below one half, or guarantee finding a violated comb whenever one has violation at least a fixed $\varepsilon\in(0,1]$ (using the scaling of Eq.~\eqref{eq:comb}). Furthermore, membership in the comb relaxation $\CP(n)$ is coNP-complete and exact optimization of arbitrary linear objectives over this relaxation is NP-hard under a polynomial-time Turing reduction.

Several directions merit further investigation. One is to determine whether hardness persists when the input is restricted to an extreme point of $\SEP(n)$. Such a result would connect the present construction more directly to the fractional points that can arise as optimal solutions of the subtour relaxation. Another direction is to refine the boundary between tractable and intractable separation. The polynomial-time solvability of the half-integral planar case~\cite{Fleischer1999} motivates studying how the complexity changes with structural parameters of the support graph. This question may identify useful classes beyond those covered by existing separation algorithms. The complexity of separating larger inequality families presents a further challenge. Extending the reduction to domino-parity, clique-tree, or path inequalities requires more than containment of the comb family. Unsatisfiable 3-SAT formulas must produce vectors satisfying every inequality in the larger family.

Finally, despite the worst-case hardness established here, very large TSP instances have already been solved to optimality~\cite{applegate2007}. There remains substantial room for further algorithmic advances through more effective heuristics, efficient exact methods for restricted classes, and their combined use within branch-and-cut solvers. We hope that our NP-completeness proof and the gadgets developed for the reduction will inform the design of new heuristics for comb separation.

\section*{Acknowledgment}
This research was financially supported by the Natural Sciences and Engineering Research Council of Canada (NSERC) and Pasqal Canada through an Alliance grant (Alliance 59081-2023). The authors acknowledge the use of OpenAI's GPT-6 and Anthropic's Claude Fable 5.1 models throughout the development of this work. Their contributions included proposing reductions and graph constructions, developing candidate proof arguments, assisting with computational checks and searches for counterexamples. The authors formulated the research question, directed the investigation, critically examined proposed arguments and individual proof steps and guided the simplification of the results. The authors retain full responsibility for the mathematical claims, the interpretation of the results and any remaining errors.

\bibliographystyle{spmpsci}
\bibliography{bib}
\end{document}